%% file: main_mpatterns.tex
\documentclass[final,onefignum,onetabnum]{siamonline220329}

\input{ex_shared}

\ifpdf
\hypersetup{
  pdftitle={Hypergraph patterns and collaboration structure},
  pdfauthor={J. L. Juul, A. R. Benson, and J. Kleinberg}
}
\fi

\usepackage{makecell}
\usepackage{color, colortbl}
\usepackage{amsfonts}
\usepackage{amsmath,amssymb}
\definecolor{Gray}{gray}{0.1}

\begin{document}


\newcommand{\ER}{Erd\H{o}s--R\' enyi{ }}
\newcommand{\mpat}{X}

\maketitle

\newcommand{\jljcom}[1]{\textcolor{orange}{[#1]}}
\begin{abstract}
Humans collaborate in different contexts such as in creative or scientific projects, in workplaces and in sports. 
Depending on the project and external circumstances, a newly formed collaboration may include people that have collaborated before in the past, and people with no collaboration history. 
Such existing relationships between team members have been reported to influence the performance of teams. 
However, it is not clear how existing relationships between team members should be quantified, and whether some relationships are more likely to occur in new collaborations than others. 
Here we introduce a new family of structural patterns, $m$-patterns, which formalize relationships between collaborators and we study the prevalence of such structures in data and a simple random-hypergraph null model. 
We analyze the frequency with which different collaboration structures appear in our null model and show how such frequencies depend on size and hyperedge density in the hypergraphs. 
Comparing the null model to data of human and non-human collaborations, we find that some collaboration structures are vastly under- and overrepresented in empirical datasets.
Finally, we find that structures of scientific collaborations on COVID-19 papers in some cases are statistically significantly different from those of non-COVID-19 papers.  Examining citation counts for $4$ different scientific fields, we also find indications that repeat collaborations are more successful for 2-author scientific publications and less successful for 3-author scientific publications as compared to other collaboration structures.
\end{abstract}



\section{Introduction}
When a new team forms, who are likely to be members of this team? 
Who are unlikely to join forces? 
Are some team constellations better suited for solving some tasks than others? 
How do external circumstances such as tight deadlines or empty schedules affect how and which teams form?

The questions above arise in all of the different settings where team formation and performance are important. 
Indeed, in online collaboration over the Web \cite{nielsen2011reinventing}, creative undertakings~\cite{guimera_team_2005}, technology and science~\cite{wu_large_2019} and school~\cite{troster_structuring_2014}, group size and the structure of social ties in the group have been reported to be of importance for the performance of teams. 
Although this diversity of settings already make the questions rich, they become even richer when one considers the plethora of external circumstances that can influence team formation in each of the settings. 
Take the COVID-19 pandemic; when researchers needed to quickly mobilize, 
analyze the spread of the disease, and its impact on society, did they work primarily in tightly-knit groups with a history of collaboration? 
Or did the interdisciplinary and high-stakes nature of the research questions make scholars work in diverse and untried teams? 

Both of the hypotheses above are reasonable and demand serious consideration. 
But how does one formalize the notion of a tightly-knit or a novel team structure? 
The essential thing to quantify in these concepts is the relationship between the members of the newly formed team. 
What were these people doing before they joined forces? 
Did subsets of the team work together before, and did others not?

Examples from popular culture richly illustrate the relevance of examining the existing relationships between team members in successful undertakings. 
For example, the American rock band Audioslave rose to popularity after being formed by Soundgarten singer Chris Cornell and $3$ former members of Rage Against the Machine: Tom Morello, Tim Commerford, and Brad Wilk.
In studio sessions, it is also common for groups of musicians to perform together repeatedly; the horn section of the legendary R{\&}B-band Tower of Power have appeared together on a large number of other artists' recordings.
In technology, the company Bumble was founded by three Tinder departees (Whitney Wolfe Herd, Chris Gulczynski and Sarah Mick) and Badoo-CEO and acquaintance of Wolfe Herd's, Andrey Andreev.  
In movies, Samuel L. Jackson stars in several Quentin Tarantino movies, and Charlotte Gainsbourg plays leading roles in $3$  of director Lars Von Trier's recent works. 

To formally study the formation of teams and existing relationships between team members, it is useful to use the language of hypergraphs.
In the hypergraph framework, people are represented by nodes, and connections \--- called hyperedges \--- can connect groups of nodes of any size that have worked together in the past. 
The focus on hypergraphs as representations of networked systems, has gained considerable traction in recent years~\cite{benson_higher-order_2016, benson_simplicial_2018,battiston_networks_2020,benson_higher-order_2021}, following two decades of intense study of graphs with only dyadic interactions~\cite{newman_networks_2018,jackson_social_2008,easley_networks_2010}.

Many of the questions being pursued in this recent work on hypergraphs are generalizations of concepts from the well-known world of dyadic interactions. 
These include questions regarding hypergraph modularity~\cite{kaminski_clustering_2019,kumar_hypergraph_2020,chodrow_generative_2021,yin_local_2017,yin_higher-order_2018,benson_tensor_2015, chien_community_2018,ahn_hypergraph_2018}, higher-order assortativity~\cite{veldt_higher-order_2021,landry2022hypergraph}, simplicial closure~\cite{benson_simplicial_2018}, hypergraph motifs and other structural patterns~\cite{lee_hypergraph_2020,kook_evolution_2020}, construction of synthetic hypergraphs with certain characteristics~\cite{courtney2016generalized,courtney2017weighted,bianconi2017emergent,young_construction_2017,kim_stochastic_2018,chodrow_configuration_2020, dyer_sampling_2020}, and how to infer higher-order network structure from data~\cite{battiston2021physics,young_hypergraph_2021}. 
The introduction of higher-order connections also makes it possible to ask completely new questions about the structure of the networked system. 
For example, a recent paper examined how hyperedges overlap in empirical hypergraphs~\cite{lee_how_2021}. 
Such a question would be trivial in the world of dyadic interactions, as dyadic interactions can only overlap in their two endpoints. In hypergraphs, however, the question is meaningful since different hyperedges could contain identical subsets of the network nodes. 

In this paper, we introduce a new family of structural patterns in hypergraphs, designed to capture the prior associations of the nodes making up a given hyperedge. 
We call these $m$-patterns, and they represent the existing relationship between groups of $m$ nodes. 
These relationships are exactly the above-mentioned quantity of interest when studying the formation of teams of size $m$. 

Formally, $m$-patterns are subhypergraphs of size $m$. 
The subhypergraph consists of the $m$ nodes under consideration, all hyperedges connecting subsets of the $m$-nodes, and fractions of hyperedges that connect subsets of the $m$-nodes to hypergraph nodes other than the $m$ under consideration. 
The inclusion of fractions of hyperedges causes $m$-patterns to quantify structure between the level of nodes and hyperedges. 
This makes $m$-patterns different from motifs and a new kind of microstructure that exists in hypergraphs, but not in graphs with dyadic interactions only. 

After having introduced $m$-patterns, we argue that the prevalence of different $m$-patterns are expected to depend on hypergraph characteristics such as hyperedge density. 
To understand this dependency, we examine how prevalence of $m$-patterns change with parameters in a $G(N,p)$-like model. 
We proceed to compare these null-model results to $m$-pattern prevalence in a wide range of datasets on human collaborations, drug networks, email networks and online tagging data.
We then examine whether collaboration structure can be influenced by external circumstances such as tight schedules. We do this by comparing collaboration structure in scientific preprints and early preprints of COVID-19 papers.
Finally, we investigate whether future citations of academic publications correlate with collaboration team structure; specifically, we compare citation counts for repeat collaborations and first-time collaborations without first-time authors.

\section{$m$-patterns in random hypergraphs}
Let us now proceed to studying past relationships between nodes in hyperedge formation. 
Our first step will be to study a simple model of random hypergraphs. 
Later, we will move from such synethetic hypergraphs and analyze node relationships in empirical hypergraphs. 
Before we can make any of these analyses, however, we must introduce the mathematical structures that we will use to understand node relationships in hyperedge formation.
\subsection{A structural pattern to summarize past relationships}
\label{sec:definition}
To define the topic of this paper, $m$-patterns, we will need some other concepts. The first of these is the notion of an induced subhypergraph~\cite{bretto2013hypergraph}.

\begin{definition}
\textbf{Induced subhypergraph}. 
The induced subhypergraph of a hypergraph $\mathcal{H}=(\mathcal{V},\mathcal{E})$ on $m$ nodes, $\mathcal{V}_I$, is a hypergraph $\mathcal{H}_I=(\mathcal{V}_I,\mathcal{E}_I)$ . For each $e\in\mathcal{E}$ that contains at least one node from $\mathcal{V}_I$, $\mathcal{E}_I$ contains a hyperedge $e'$  linking all nodes that are both in $e$ and $\mathcal{V}_I$. 
\label{def:induced_hypergraph}
\end{definition}

It is clear that an induced subhypergraph completely summarizes all existing relationships between its constituting nodes. The final sentence of Definition~\ref{def:induced_hypergraph} means that $\mathcal{H}_I$ contains fractions of the hyperedges of $\mathcal{H}$. This makes the induced subhypergraph an interesting object for hypergraphs. For graphs, fractions of edges are simply vertices, and so the graph equivalent of this definition would just be a subgraph on $m$ chosen nodes. 
If we do not need the entire relationship history between nodes, but are content with summarizing the largest subsets of nodes that have collaborated in the past, the following definition is useful.

\begin{definition}
\textbf{Maximal induced subhypergraph}. 
The maximal induced subhypergraph $\mathcal{H}_I=(\mathcal{V}_I,\mathcal{E}_I)$ of a hypergraph $\mathcal{H}=(\mathcal{V},\mathcal{E})$ on $m$ nodes, $\mathcal{V}_I$, is the corresponding induced subhypergraph made simple by removing all hyperedges from $\mathcal{E}_I$ that are entirely contained in other hyperedges in $\mathcal{E}_I$.
\end{definition}

The key difference between an induced subhypergraph and a maximal induced subhypergraph is that the latter is simple. A simple hypergraph is defined as follows~\cite{berge1984hypergraphs}.
\begin{definition}
\textbf{Simple hypergraph}. 
A hypergraph $\mathcal{H}=(\mathcal{V},\mathcal{E})$ is simple if none of its hyperedges are entirely contained in another, $\nexists s_i,s_j\in\mathcal{E}:s_i\subseteq s_j$.
\label{def:simple_hypergraph}
\end{definition}
Notice that simple hypergraphs are different from simple graphs in that simple hypergraphs can contain self-looping hyperedges. We note that the hypergraphs we consider in this paper generally are not simple. Simple hypergraphs play a different role in this story. 
Because simple hypergraphs cannot have parallel edges there exists only a finite number of different such hypergraphs of size $m$. This is a nice feature if we are interested in quantifying typical relationship structures among people that choose to form teams. This is exactly what we are interested in, so we refer to these finitely many relationship structures on $m$ nodes as {\em $m$-patterns}.

\begin{definition}
\textbf{$m$-pattern}. 
A simple hypergraph with $m$ vertices is an $m$-pattern.
\label{def:m-patterns}
\end{definition}

With the concept of an $m$-pattern in hand, we are now ready to look for instances of $m$-patterns in larger hypergraphs. 

\begin{definition}
\textbf{Instance of an $m$-pattern}. 
An instance of an $m$-pattern $\mpat$ in the hypergraph $\mathcal{H}=(\mathcal{V},\mathcal{E})$ is a maximal induced subhypergraph $\mpat’$ on $m$ nodes  which is isomorphic to $\mpat$.
\label{def:instance_of_m-pattern}
\end{definition}

Figure~\ref{fig:m-pattern_illustration} illustrates what such $m$-patterns from maximally induced subhypergraphs might look like. The figure shows three collaborations. Some people in these collaborations have worked together previously \--- perhaps in larger groups. Such larger collaborations become $k$-node hyperedges in the $m$-patterns that the collaboration structure form.

\begin{figure}
    \centering
    \includegraphics[width=.5\linewidth]{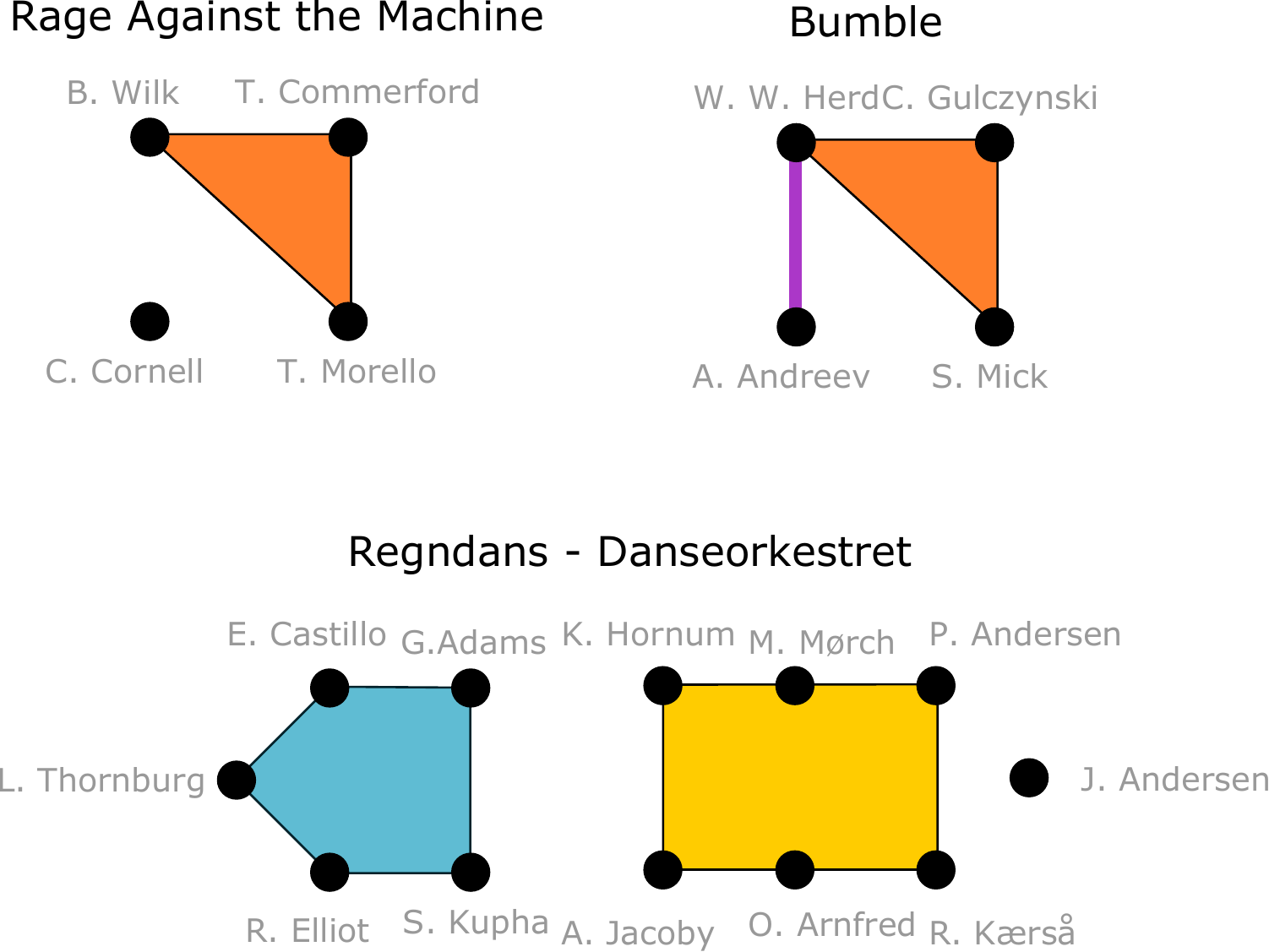}
    \caption{Illustration of the relationship between individual members of Audioslave, Bumble founders and musicians on recording of Regndans by Danseorkestret.
    \label{fig:m-pattern_illustration}
    \vspace{-15pt}
    }
\end{figure}

With the definition of $m$-patterns, and their instances in hypergraphs, we now have a formal way of talking about  
existing relationships between hypergraph vertices.
In particular, 
when a new team of $m$ individuals appears,
we consider the team members' past history of interactions to be the $m$-pattern consisting of all maximal subsets that have worked together before.
\begin{definition}
\textbf{Instance of a labelled $m$-pattern}. 
An instance of a labelled $m$-pattern $\mpat$ with assigned vertex labels $1,2\ldots m$ in the hypergraph $\mathcal{H}=(\mathcal{V},\mathcal{E})$ is a maximal induced subhypergraph $\mpat’$ on $m$ nodes with assigned vertex labels $1,2,\ldots, m$ which is isomorphic to $\mpat$ and where corresponding vertices have the same assigned labels as in  $\mpat$.
\label{def:instance_of_labelled_m-pattern}
\end{definition}

In Appendix Section~\ref{app:sec:hypergraph_definitions_illustration}, we illustrate connections between some of the concepts introduced in this section.

We are now ready to examine what $m$-patterns among nodes precede hyperedge formation in hypergraphs. 
In the following subsection, we will do so in a class of synthetic random hypergraphs. 

\subsection{$G^{(m)}(N,p)$ model of random hypergraphs}
\label{sec:G^m(N,p)}
From Definition~\ref{def:m-patterns} and~\ref{def:instance_of_m-pattern}, it is clear that the structure of the underlying hypergraph $\mathcal{H}$ greatly influences what $m$-patterns that can exist among sets of $m$ nodes, and what multiplicity these $m$-patterns might have in the hypergraph. 
If the hypergraph is very sparse, most sets of $m$ nodes have never collaborated before. 
For sparse hypergraphs, this presents us with the following question. 
When a newly formed team consists of $m$ nodes with no past collaborations, is this because people tend to team up with strangers, or because the underlying hypergraph is sparse? 
If we want to understand whether some existing relationship structures are more likely to give rise to future team formations, we must know what to expect by chance alone. 
Studying $m$-patterns in a null-model of hypergraphs can help us gain intuition about what $m$-patterns we should expect to dominate at different hyperedge densities. 

We choose to study $m$-patterns in a hypergraph generalization of the widely-studied random-graph family known as \ER graphs \--- or $G(N,p)$.  $G(N,p)$ is known to create unrealistically simple graph structures. Nonetheless, the dyadic $G(n,p)$ model has been a major driver in the development of the study of networks: it is the simplest random-graph model, analytically tractable, and its phenomena are correspondingly clear to articulate. We study a $G(N,p)$-type model for the same reasons.

In the classic $G(N,p)$ model, an $N$-vertex random graph is created by inserting each possible edge with probability $p$~\cite{newman_networks_2018}. 
Various hypergraph generalizations of the $G(N,p)$ model have been studied in the past~\cite{linial2006homological, meshulam2009homological, costa2016random, kahle2009topology, kahle2014topology}. 
We choose to study a version where a hypergraph with $N$ nodes and $m$-vertex hyperedges is created by inserting each possible hyperedge connecting $m$ nodes with probability $p$. 
Since the parameters $N$, $p$ and $m$ define this hypergraph family, $G^{(m)}(N,p)$ is a natural name to summarize the family. 
The dyadic \ER graphs, normally known as $G(N,p)$, would be $G^{(2)}(N,p)$ in this notation. 

With the $G^{(m)}(N,p)$ model in hand, we set out to examine how often a new hyperedge would join $m$ nodes with $m$-pattern $\mpat$ by chance, given that the hyperedge is forming in a hypergraph created using the $G^{(m)}(N,p)$ model with parameters $N$, $p$ and $m$. 
To quantify this, we create a large number of $G^{(m)}(N,p)$ hypergraphs and count the average fraction of sets of $m$ nodes that form each pattern $\mpat$ across these many random hypergraphs for choices of hypergraph size $N$, set size $m$ and as a function of hyperedge probability $p$. 
In Figure~\ref{fig:random_m-patterns}, we show results obtained for two such simulations. 
In Figure~\ref{fig:random_m-patterns}A, the constructed hypergraphs have size $N=50$ and contain hyperedges joining $m=3$ nodes. 
In Figure~\ref{fig:random_m-patterns}B, the hypergraphs have size $N=100$ and hyperedges join $m=4$ nodes. 
The first thing to notice about these figures is that, when increasing $p$ from $0$, all but two $m$-patterns increase in prevalence, experience peak prevalence, and finally become less common again. 
The two patterns that do not take such journeys are: 1) the pattern in which noone collaborated with anyone before; and 2) the repeat collaboration. 
The occurrences of the no-past-collaboration pattern monotonically declines with $p$, whereas the repeat-collaboration pattern increases monotonically with $p$. These ``exceptions'' are easily understood: As $p$ increases, more nodes become part of $m$-node hyperedges. A higher $p$ means that fewer nodes avoid collaborations altogether, whereas $m$-node collaborations (what we also call repeat collaborations) increase linearly with $p$.  

Having noticed regularities in the general shape of prevalence curves in Fig.~\ref{fig:random_m-patterns}A, another interesting observation is that not all patterns get to be the most common for any $p$ in Fig.~\ref{fig:random_m-patterns}B. 
For example, the pattern consisting of a single $3$-node hyperedge and a solitary node (dashed orange line) never outgrows all other patterns. 
This observation is interesting enough that we introduce a term for a pattern which gets to be the most common at a given value of $p$.
\begin{definition}
\textbf{Extreme pattern} An $m$-pattern, $\mpat$, is extreme if, for a particular value of $N$, the $m$-pattern is the most prevalent of all $m$-patterns for some $p$.
\label{def:extreme}
\end{definition}

\begin{definition}
\textbf{Extreme in the limit} An $m$-pattern, $\mpat$, is extreme in the limit if there exists an $N_0$ such that for all $N>N_0$ there exists a $p$ where the pattern is the most prevalent of all $m$-patterns in the hypergraph.
\end{definition}

A third interesting observation from Figure~\ref{fig:random_m-patterns} is the order in which extreme patterns are the most common in the hypergraphs when increasing $p$. 
As $p$ increases, the pattern with no previous collaborations is the most common at first. 
Then follow patterns containing disjoint nodes that all have previous collaborations, but none with each other. 
Then patterns that include dyadic collaborations, etc. 
These observations beg for explanations. 
Can we understand the shape of the prevalence curves and estimate them analytically?
Can we understand which $m$-patterns are extreme and for which hyperedge densities, $p$, these patterns are the most common? 

\begin{figure}
    \centering
\includegraphics[width=\linewidth]{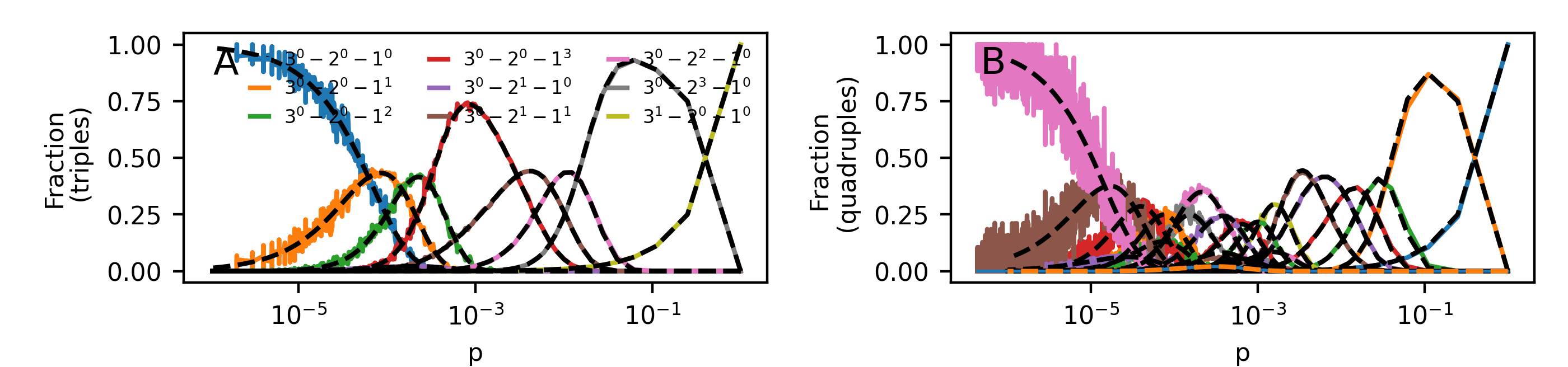}
    \caption{Frequency of $m$-patterns in the $G^{(m)}(N,p)$ model as a function of p for $m=3$, $N=50$ (\textbf{A}) and $m=4$, $N=100$ (\textbf{B}). Each datapoint is the average frequency of an $m$-pattern in $100$ independent simulations of the model. The pattern $3^i-2^j-1^k$ contains $i$ 3-node, $j$ $2$-node and $k$ $1$-node hyperedges. Analytical estimates of prevalence in Eq.~\eqref{eq:analytical_expression} is plotted with dashed lines. In (B), multiple curves are plotted with same colors;  See Appendix Section~\ref{app:sec:labelled2B} for a labelled version of B. 
    \label{fig:random_m-patterns}}
\end{figure}

The answers to both of the above questions are yes. 
With the following theorem, we identify a sizeable number of patterns that cannot be extreme in the limit.
\begin{theorem}
If the pattern $\mpat$ contains $H$-node hyperedges and misses $l+1$-node hyperedges, and $|H-l|\ge 2$, $\mpat$ is not extreme in the limit.
\label{thm:extreme-patterns}
\end{theorem}

Theorem~\ref{thm:extreme-patterns} tells us why the $m$-pattern with a 3-node hyperedge and a solitary node is not extreme in Figure~\ref{fig:random_m-patterns}B (or rather, why it would not be in the limit $N\to \infty$). 
The reason is that the pattern contains a $3$-node hyperedge, and misses $2$-node hyperedges that could have existed. Since $|3-(2-1)|=2$, 
 Theorem~\ref{thm:extreme-patterns} tells us that such a pattern cannot be extreme in the limit.

In order to prove Theorem~\ref{thm:extreme-patterns}, we will need $2$ Lemmas. 
The first Lemma conveniently answers the second question we asked above: Can we understand the shape of the prevalence curves of $m$-patterns?
We will answer this question by writing down a formula for the expected frequency of the $m$-pattern $\mpat$ among the instances of $m$-patterns in $G^{(m)}(N,p)$ hypergraphs.
We can do this if we think about the prevalence of an $m$-pattern in the following way. 
The fraction of sets of $m$ nodes that form an $m$-pattern $\mpat$ in a $G^{(m)}(N,p)$ hypergraph is equal to the probability that the pattern is formed by the $m$ nodes when each size-$m$ hyperedge is inserted with probability $p$. 
Calculating this probability is an exercise in combinatorics. 
The result reveals that the prevalence curve of any $m$-pattern takes the same analytical form.
\begin{lemma}
Let $\mpat$ be a pattern consisting of $x_m$ $m$-node hyperedges, $x_{m-1}$ $(m-1)$-node hyperedges, ... , and $x_1$ $1$-node hyperedges. In addition, denote the number of missing $i$-node hyperedges by $y_i(x_i,x_{i+1},\ldots,x_m)$. For $N\ge m$ nodes and $0\le p \le 1$, the prevalence of $\mpat$, can be written,
\begin{equation}
    P(\mpat) = \gamma_\mpat\prod_{i=1}^m p_i^{x_i}(1-p_i)^{y_i(x_{i}, x_{i+1},...,x_m)}.
    \label{eq:analytical_expression}
\end{equation}
Here, $\gamma_{\mpat}\in\mathbb{N}$ is a combinatorial factor and $p_i$ is the probability that $i$ nodes chosen uniformly at random from the $N$ nodes, are connected by an $i$-simplex,
\begin{equation}
    1-p_i = (1-p)^{c_i},
    \label{eq:p_i}
\end{equation} 
where we defined ${N-m \choose m-i} = c_i$.
\label{lem:mpat_analytical}
\end{lemma}
The combinatorial factor $\gamma_\mpat$ counts the number of isomorphic configurations of $\mpat$ that exists on $m$ nodes. 
Hence, the prevalence curve of a labelled version of $m$-patterns can be obtained by setting $\gamma_{\mpat}=1$. 
A side-effect of this fact is that all labelled versions of an $m$-pattern are equally likely under the $G^{(m)}(N,p)$ model. 
\begin{lemma}
\label{lem:zeroone}
For any $\epsilon>0$ and large enough $N$, the values of $p$ at which $p_l=a$, for $0<a<1$, $p_k$ take the values
\begin{equation}
    p_k
    \begin{cases}
        >1-\epsilon, & \text{ if } k\le l-1, \\
        < \epsilon, & \text{ if }k \ge l+1.
    \end{cases}
\end{equation}
\end{lemma}
\begin{proof}
If $p_l=a$, Lemma~\ref{lem:mpat_analytical} allows us to find the corresponding value of $p$,
\begin{align}
    p &= 1-(1-a)^{1/c_l}.
\end{align}
Inserting this in the formula for $p_k$ gives us
\begin{equation}
    p_k = 1-(1-a)^{c_k/c_l}.
\end{equation}
Now let $k>l$. Then,
\begin{align}
    \frac{c_l}{c_k} &=\frac{\binom{N-m}{m-l}}{\binom{N-m}{m-k}},\\
    &\ge(N-m)^{k-l}\frac{(m-k)!}{(m-l)^{m-l}},\\
    &=(N-m)^{k-l}\beta^{-1}.
    \label{eq:firstbound_reciprocal}
\end{align}
Here we used the inequality
\begin{equation}
    \frac{n^k}{k^k}\le \binom{n}{k} \le \frac{n^k}{k!}
    \label{eq:inequalities}
\end{equation}
repeatedly and defined $\beta = (m-l)^{m-l}/(m-k)!$. Taking the reciprocal value of both sides of Eq.~\eqref{eq:firstbound_reciprocal} gives us the bound
\begin{equation}
    \frac{c_k}{c_l}\le (N-m)^{l-k}\beta.
\end{equation}
Because $0<(1-a)<1 $,
\begin{align}
    p_k &= 1-(1-a)^{c_k/c_l} \\
    &\le 1-(1-a)^{\beta(N-m)^{l-k}}.
\end{align}
What does $N$ need to be larger than, if $p_k<\epsilon$? Demanding that
\begin{equation}
    1-(1-a)^{\beta(N-m)^{l-k}}<\epsilon,
\end{equation}
ensures that $p_k<\epsilon$ and allows us to isolate $N$,
\begin{equation}
    N >m+\left[\beta \frac{\ln \left( 1-a \right)}{\ln \left( 1-\epsilon \right)}  \right]^{1/(k-l)}.
        \label{eq:N1}
\end{equation}
This proves half of the Lemma. For the other half, we now let $k<l$. With similar steps as in the previous case, we can get the bound,
\begin{equation}
    \frac{c_l}{c_k}\le \beta'(N-m)^{k-l},
\end{equation}
with $\beta' = (m-k)^{m-k}/(m-l)!$. Taking the reciprocal value of both sides, the bound becomes,
\begin{equation}
        \frac{c_k}{c_l}\ge \beta'^{-1}(N-m)^{l-k}.
\end{equation}
We now proceed in analogous manner as in the first half of the proof. With the bound on $c_k/c_l$,
\begin{align}
    p_k &= 1-(1-a)^{c_k/c_l} \\
    &\ge 1-(1-a)^{\beta'^{-1}(N-m)^{l-k}}.
\end{align}
If this final quantity is larger than $1-\epsilon$, $p_k$ is too. For what $N$ is this the case then? Setting the final expression larger than $1-\epsilon$ and isolating $N$ yields
\begin{equation}
    N>m+\left[\beta' \frac{\ln \epsilon}{\ln (1-a)} \right]^{1/l-k}.
    \label{eq:N2}
\end{equation}
We conclude that if $N$ is larger than both of the values given in Eqs.~\eqref{eq:N1} and \eqref{eq:N2}, 
\begin{equation}
        p_k
    \begin{cases}
        >1-\epsilon, & \text{ if } k\le l-1, \\
        < \epsilon, & \text{ if }k \ge l+1.
    \end{cases}
\end{equation}
This proves the Lemma.
\end{proof}
With Lemmas~\ref{lem:mpat_analytical} and~\ref{lem:zeroone}, we now present our proof of Theorem~\ref{thm:extreme-patterns}.
\begin{proof}
\textbf{(Theorem~\ref{thm:extreme-patterns})} If the pattern $\mpat$ is extreme, all factors in the analytical expression for its prevalence must be large enough that $P(\mpat)$ takes a larger value than $P(\mpat')$ for any other pattern $\mpat'$. By Lemma~\ref{lem:mpat_analytical}, $P(\mpat)$ contains factors $(1-p_{l+1})^{y_{l+1}}$ and $p_{H}^{x_H}$, with $y_{l+1},x_H\ne0$. By Lemma~\ref{lem:zeroone}, if for some $p$, $p_H$ takes a value bounded away from 0 and 1, then one can choose an $N$ large enough to make $p_k$ arbitrarily close to $1$, if $k\le H-1$. For any such $k$, $(1-p_k)$ then becomes arbitrarily close to $0$. Hence, if $P(\mpat)$ contains factors of both $p_H$ and $(1-p_{l+1})$, and $|H-l|\ge2$, $P(\mpat)\to 0$ for large enough $N$, which implies it cannot be extreme in the limit.
\end{proof}

Theorem~\ref{thm:extreme-patterns} settles that a large class of $m$-patterns cannot be extreme in the limit. 
A natural next question to ask is then, what patterns are extreme in the limit? 
Are some types of patterns bound to be extreme? 
Are some types of patterns only extreme for certain choices of $m$?

Proving such positive results appears to be more challenging than proving the negative results of Theorem~\ref{thm:extreme-patterns}. 
A useful concept in proving such positive results is what we call a \textit{pure} pattern.
\begin{definition}
\textbf{(Pure pattern)}
An $m$-pattern with no hyperedges other than all possible $k$-node hyperedges is a pure pattern. \label{def:pure_pattern}
\end{definition}
Pure patterns are easy to think about and work with because they contain only one kind of hyperedge, and there is only a single way of constructing each pure pattern. 
In Lemma~\ref{lem:mpat_analytical}, this means that $\gamma_\mpat = 1$ for any pure pattern. 
The simplicity of working with pure patterns has caused these patterns to play a central role in our results on which patterns are actually extreme. 
One important result concerns exactly these pure patterns (proof given in Appendix~\ref{app:proof:thm:pure_patterns}).
\begin{theorem}
All pure patterns are extreme in the limit.
\label{thm:pure_patterns}
\end{theorem}

Our next theorem requires a result for labelled $m$-patterns. We remind the reader that instances of labelled $m$-patterns are different from instances of $m$-patterns in that we do not group isomorphic maximal induced subhypergraphs together. In this case, Lemma~\ref{lem:mpat_analytical} still gives us the analytical expression for the prevalence of labelled $m$-patterns, but $\gamma = 1$ for all patterns.

The following two lemmas are proven in Appendices~\ref{app:proof:lem:pure_dominates_above_1/2} and 
\ref{app:proof:lem:pure_dominates_below_1/2}.

\begin{lemma}
For labelled $m$-patterns  and $N\to\infty$, when $p_k>\frac{1}{2}$, the pure pattern containing only $k$-node hyperedges is more frequent than all patterns consisting of both $k$-node hyperedges and $(k-1)$-node hyperedges.
\label{lem:pure_dominates_above_1/2}
\end{lemma}

\begin{lemma}
For labelled $m$-patterns and $N\to\infty$, when $0<p_{k+1}<\frac{1}{2}$, the pure pattern containing only $k$-node hyperedges is more frequent than all patterns consisting of both $(k+1)$-node hyperedges and $k$-node hyperedges.
\label{lem:pure_dominates_below_1/2}
\end{lemma}

\newcommand{\omt}[1]{}
\omt{
\begin{proof}
By Lemma~\ref{lem:zeroone}, if $p_{k+1}$ is non-zero, we can choose $N$ large enough $p_l\to1$ as $N\to \infty$ for $l\le k-1$. 
The prevalence of the pure pattern with $k$-node hyperedges is 
\begin{equation}
    P(\mpat_k) = p_k^{{m \choose k}}(1-p_{k+1})^{{m\choose k+1}}\Omega',
    \label{eq:kpure_elab}    
\end{equation}
where $\Omega'=\prod_{i=k+2}^m (1-p_i)^{{m\choose i}}$. 
The prevalence of the non-pure patterns containing $x_k$ $k$-node and $x_{k+1}$ $(k+1)$-node hyperedges is,
\begin{equation}
    P(\mpat_{k-1,k}') = p_{k}^{x_{k}} p_{k+1}^{x_{k+1}} (1-p_{k+1})^{{m\choose k+1}-x_{k+1}} \Omega'.
    \label{eq:k+1_nonpure}
\end{equation}
By Lemma~\ref{lem:zeroone}, for any $\epsilon>0$ and large enough $N$, the first factor in both Eqs~\eqref{eq:kpure_elab} and \eqref{eq:k+1_nonpure} can get arbitrarily close to $1$. 
Hence, the pure and non-pure patterns above cross when the remaining factors are equal in Eqs~\eqref{eq:kpure_elab} and \eqref{eq:k+1_nonpure}. This happens when $p_{k+1}=\frac{1}{2}$. 
For lower values of $p_{k+1}$, $(1-p_{k+1})>p_{k+1}$. 
This proves the lemma.
\end{proof}
}

These two Lemmas and Theorem~\ref{thm:extreme-patterns} give us the following interesting result.
\begin{theorem}
For labelled $m$-patterns, only pure patterns are extreme in the limit.
\label{thm:labelled_extreme}
\end{theorem}
Moreover, the arguments leading to Lemmas~\ref{lem:pure_dominates_above_1/2} and~\ref{lem:pure_dominates_below_1/2} also lead us to the following Lemma  (see also Appendix~\ref{app:illustration:thm:non-pure}),
\begin{lemma}
For labelled $m$-patterns, all patterns consisting only of $(k+1)$-node hyperedges and all possible remaining $k$-node hyperedges are equally prevalent when $p_{k+1}=\frac{1}{2}$.
\label{lem:cross_1/2}
\end{lemma}
These results for labelled $m$-patterns help us prove the following more general theorem for non-labelled patterns.
\begin{theorem}
If $m\ge 3$ at least one non-pure pattern is extreme.
\label{thm:non-pure}
\end{theorem}
\begin{proof}
If $m\ge 3$, non-pure patterns exist that do not violate Theorem~\ref{thm:extreme-patterns}. 
Since we are not dealing with labelled patterns, the combinatorial factor $\gamma$ is some integer larger than or equal to $1$ for each pattern. 
For pure patterns $\gamma_\mpat = 1$. 
Now focus at the point $p_{m-1}=\frac{1}{2}$. 
From Lemma~\ref{lem:cross_1/2}, prevalence curves for several pure and non-pure \textit{labelled} patterns cross at this point in the large-$N$ limit. 
At least one of the corresponding non-labelled non-pure patterns has $\gamma_\mpat \ge 2$. 
For example, the pattern missing a single $p_{m-1}$-node hyperedge and containing ${m-1 \choose m-2}$ $(m-2)$-node hyperedges instead has $\gamma_\mpat = m-1$. 
Hence, in this point, at least this non-pure pattern is more prevalent than the two pure patterns containing $(m-2)$-node and $(m-1)$-node hyperedges. 
For this reason, and Lemma~\ref{lem:zeroone}, it is more prevalent than all pure patterns. 
This proves the Theorem.
\end{proof}

Having shown that all pure patterns are extreme and that some none-pure patterns are extreme, too, we present a final result that shows that a large number of potentially extreme patterns are not extreme (proof in Appendix~\ref{app:proof:thm:two_patterns_not_extreme}). 
\begin{theorem}
Two different $m$-patterns that have different combinatorial factors and consist only of $x_k$ $k$-node hyperedges and all possible remaining $(k-1)$ hyperedges cannot both be extreme in the limit.
\label{thm:two_patterns_not_extreme}
\end{theorem}
We note that in cases where several patterns compete for being extreme as described in Theorem~\ref{thm:two_patterns_not_extreme}, the pattern that actually gets to be extreme in the limit can have very different structure depending on $m$. 
The reason for this is that the combinatorial factor $\gamma$ depends on the value of $m$.
Take for example the two possible non-isomorphic patterns consisting of two two-node hyperedges and all remaining possible one-node hyperedges for $m\ge4$.
In one pattern the two $2$-node hyperedges share a node, whereas in the other, the $2$-node hyperedges are completely separate. 
For a given choice of $m$, there are $3{m \choose 3}$ ways of constructing the $m$-pattern with linked $2$-node hyperedges, and 3${m \choose 4}$ ways of constructing the pattern with separate $2$-node hyperedges. 
Hence, patterns with $2$-node hyperedges in sequence have larger combinatorial factors when $m\le 6$, the patterns have the same combinatorial factor if $m=7$ and patterns with parallel $2$-node hyperedges dominate when $m\ge 8$.

\section{Hypergraph patterns in empirical data}
The $G^{(m)}(N,p)$ model informs us how prevalent we should expect an $m$-pattern $\mpat$ to be in an $N$-node hypergraphs where a fraction $p$ of possible $m$-node hyperedges exist if the hyperedges were distributed uniformly randomly among all possible $m$-node hyperedges. This raises a natural question: 
In empirical datasets, are some $m$-patterns overrepresented and others underrepresented 
compared to the $G^{(m)}(N,p)$ null-model? 

\subsection{Academic coauthorship hypergraphs}
Making an informative comparison of $m$-patterns in empirical hypergraphs and the $G^{(m)}(N,p)$ model is not as straight forward as it sounds. Any empirical hypergraph has a fixed number of nodes and a given hyperedge density. For this reason, any comparison of the $G^{(m)}(N,p)$ model to an empirical hypergraph results in a comparison for just one value of $p$. Since one of the interesting features of the $G^{(m)}(N,p)$ model is how the prevalence of the $m$-patterns change with the hyperedge density $p$, we seek a large collection of hypergraphs with different hyperedge densities. We construct such a collection from the set of ego hypergraphs in empirical coauthorship hypergraphs. For each node $v$ in the coauthorship hypergraph, $\mathcal{H} = (\mathcal{V},\mathcal{E})$, we construct an ego hypergraph $\mathcal{H}_e=(\mathcal{V}_e,\mathcal{E}_e)$. $\mathcal{V}_e$ includes all neighbors of $v$, but not $v$ itself. $\mathcal{E}_e$ includes all $m$-node hyperedges between nodes in $\mathcal{V}_e$.  Furthermore, for any $m'$-node hyperedge ($m'\ge m+1$) in $\mathcal{E}$ that joins $m$ nodes from $V_e$ and ($m'-m$) nodes from $\mathcal{V}\setminus (\mathcal{V}_e \cup  v)$, we include a subhyperedge in $\mathcal{V}_e$ joining these nodes from $\mathcal{V}_e$.

Figure~\ref{fig:human_collab}A shows the prevalence of $3$-patterns in ego hypergraphs in a coauthorship hypergraph of scientists working in the field of Geology~\cite{benson_simplicial_2018}. These ego hypergraph have very diverse hyperedge densities, $p$ (horizontal axis). The ego hypergraphs also have different sizes, $N$. In the plot, we include results for all ego hypergraphs of sizes $10 \le N \le 50$. Since the prevalence of $m$-patterns depends on $N$ in the $G^{(m)}(N,p)$ model, the data points are not expected to fall on clear lines as were found for the null model. Indeed, instead of lines, datapoints for each pattern form point clouds in the Figure. This makes it difficult to compare the data to the model.

In Figure~\ref{fig:human_collab}B, we show the same data after performing a rolling average. In this panel we split the logarithmic horizontal axis into equidistant segments; $10$ for each order of magnitude. For each segment, we calculate an average prevalence of all $3$-patterns $\mpat$. Every datapoint with $p$-value between the $p$-values of segments $i-1$ and $i+1$ count in the average calculated for segment $i$. The data is plotted with dots. The $G^{(m)}(N,p)$ expectation (curves) was created by plugging the empirical values for $N$ and $p$ for each ego hypergraph into the $G^{(m)}(N,p)$ model. We then performed our averaging procedure to the resulting point cloud. 

Although there are similarities between prevalence curves of $3$-patterns in the empirical ego hypergraphs and the model, there are clear discrepancies as well. For example, the pattern with just a single $1$-hyperedge is clearly overrepresented in the data for several orders of magnitude of the hyperedge density $p$. On the other hand, the pattern consisting of a $1$-node and a $2$-node hyperedge is underrepresented in the data. Similar plots of a dataset of coauthorships in the field of history confirms these observations (Figure~\ref{fig:human_collab}C,D).
\begin{figure}[tb]
    \centering
    \includegraphics[width=\linewidth]{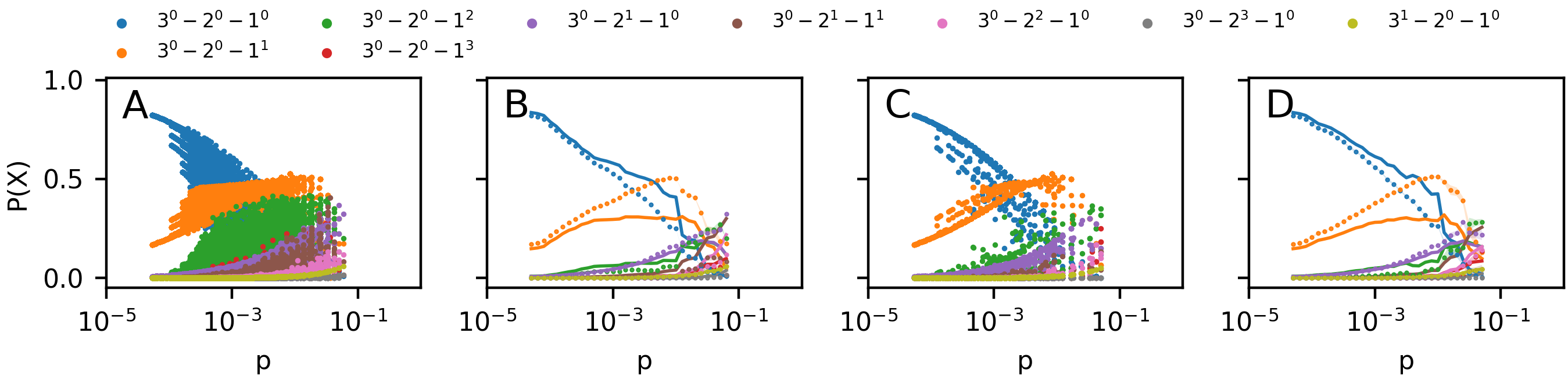}
    \caption{\textbf{A} Frequency of $m$-patterns in ego networks for sizes $10\le N\le 50$ in the Geology coauthorship network~\cite{benson_simplicial_2018}. \textbf{B} Rolling average of data in (A) plotted alongside $G^{(m)}(N,p)$ prediction (curves) \textbf{C} As in (B) but for a History coauthorship network~\cite{benson_simplicial_2018}  \textbf{D} As in (B). In all panels, colors indicate the $m$-pattern shown in the legend. 
    }
    \label{fig:human_collab}
\end{figure}
\subsection{Hypergraphs of human and non-human systems}
The similarity of Figure~\ref{fig:human_collab}B and Figure~\ref{fig:human_collab}D is striking. For the two different coauthorship hypergraphs, many of the same patterns seem to be underrepresented and overrepresented as compared to the $G^{(m)}(N,p)$ null model. The two datasets both stem from academic coauthorship hypergraphs. Could the similarities in $m$-pattern prevalence be due to the fact that the hypergraphs stem from the same domain? And if so, which patterns are overrepresented or underrepresented in hypergraphs from other domains?

\begin{figure}[!tb]
    \centering
    \includegraphics[width=\linewidth]{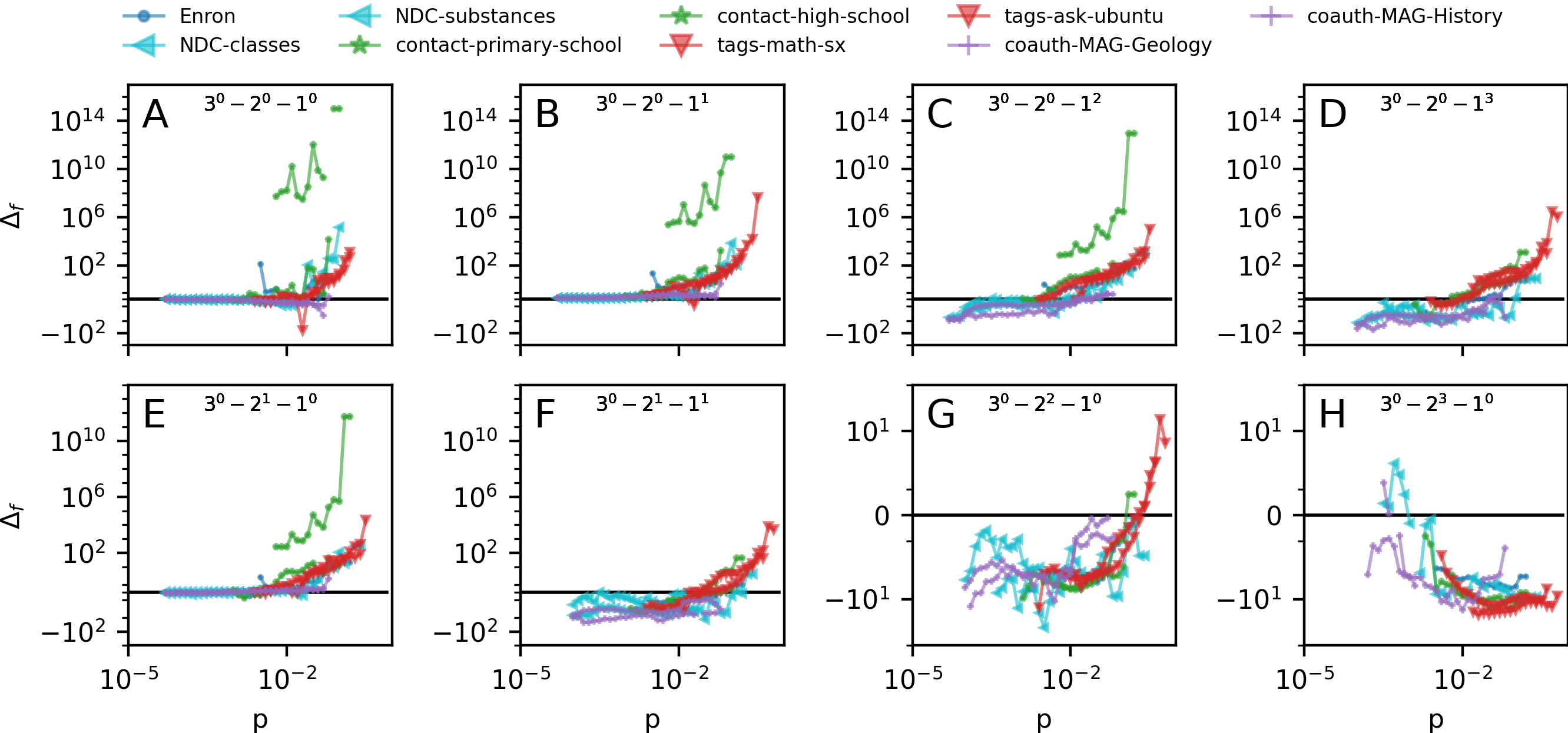}
    \caption{\textbf{A-H} quantify the difference between $m$-pattern prevalence in $9$ datasets as compared to our $G^{(m)}(N,p)$ model. The difference measure is $\Delta_f = [P(\mpat_{\rm data})-P(\mpat_{\rm model})]/\min(P(\mpat_{\rm model}),P(\mpat_{\rm data}))$ and $\Delta_f=0$ corresponds to perfect agreement between data and model. Each panel plots $\Delta_f$ for a specific $3$-pattern as found in all $9$ datasets.}
    \label{fig:diff_data_model}
\end{figure}
In Figure~\ref{fig:diff_data_model} we compare the prevalence of $m$-patterns in ego hypergraphs of 9 different empirical hypergraphs to the $G(N,p)$ model. The hypergraphs represent very different domains: Human and non-human, processes on the web and in nature. Hypergraphs represent email networks (``Enron''), drug networks (``NDC-classes'' and ``NDC-substances''), human contact networks (``contact-primary-school'' and ``contact-high-school''), online tagging data (``tags-math-sx'' and ``tags-ask-ubuntu'') and the academic coauthorship networks introduced above. The vertical axes quantify the difference between the prevalence of $m$-patterns in the empirical ego hypergraphs and the $G(N,p)$ model, $\Delta_f = [P(\mpat_{\rm data})-P(\mpat_{\rm model})]/\min(P(\mpat_{\rm model}),P(\mpat_{\rm data}))$. The color and shape of the marker depends on the domain that the ego hypergraph represents.

The first thing to notice in Figure~\ref{fig:diff_data_model} is how numerically large the values on the vertical axes are (note the symmetrical logarithmic axes). If a datapoint is plotted at vertical value $10$, the pattern is $10$ times more prevalent in the data than in the model. So with the vertical scale covering the interval $[-10^2,10^{14}]$, some patterns are vastly over and underrepresented in the data. 

A second thing to notice in Figure~\ref{fig:diff_data_model} is that some patterns are consequently underrepresented in data. Most clearly underrepresented is the pure pattern of $2$-node hyperedges (Figure~\ref{fig:diff_data_model}H). For all datasets but ``NDC-classes'' this lies clearly in the negative vertical values. The pattern consisting of a $2$-node and $1$-node hyperedge and the pattern with just $2$ $2$-node hyperedges (Figure~\ref{fig:diff_data_model}F and G) are also mostly underrepresented in the datasets. 

A third and interesting aspect of Figure~\ref{fig:diff_data_model} is hints of similarities between datasets from similar domains. With the exception of the school contact networks, datapoints from similar domains fall very close together on the plots.
\begin{figure}[!tb]
    \centering
    \includegraphics[width=\linewidth]{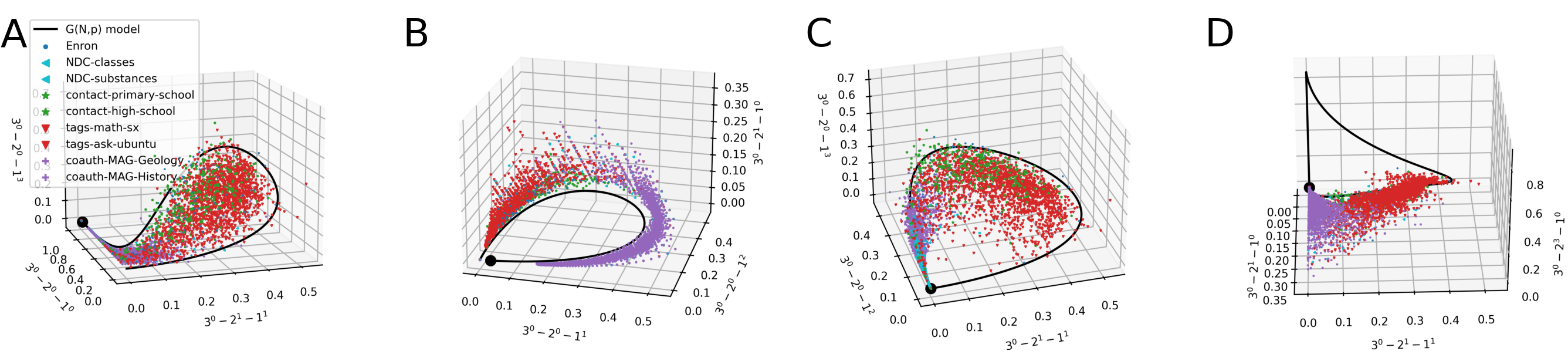}
    \caption{\textbf{A-D} Scatter plots of ego networks in $9$ empirical datasets. Markers as in Figure~\ref{fig:diff_data_model}. Each axis is a $3$-pattern; axes are different in panels. The $G^{(m)}(N,p)$ model traces out the black curves; the black dot corresponds to the lowest hyperedge density $p$ on the curve. }
    \label{fig:3dscatter}
\end{figure}
Figure~\ref{fig:diff_data_model} represents one way of comparing prevalence of $m$-patterns for different datasets. In Figure~\ref{fig:3dscatter} we provide another. Each panel in the figure shows a scatter plot of the prevalence of $3$ $3$-patterns in each of the empirical ego hypergraphs. The color and shape of the marker depends on the domain that the ego hypergraph represents. We also plot the results for our $G^{(m)}(N,p)$ model (with $N=50$). In all panels, the model traces out a parametric curve starting in the point marked by a black dot. Interestingly, the data are not scattered all around the curve; instead, for these scatter plots, datapoints often fall in a limited subspace around the curve. Lastly, the panels show that datapoints from similar domains fall close together. We note that some of this separation could be due to the different orders of magnitude of the hyperedge densities, $p$, present in each dataset (see Figure~\ref{fig:diff_data_model}).

\section{COVID-19 collaborations}
In the previous $2$ sections, we have counted the prevalence of $m$-patterns in empirical ego hypergraphs and our $G^{(m)}(N,p)$ model. The hypergraphs we were examining were always fully grown. One of our main motivations for introducing $m$-patterns was to investigate what prior relationships between a set of $m$ nodes are likely to exist when these nodes choose to collaborate. To confront this question, we now examine hyperedge formation in a growing hypergraph: the coauthorship network of papers submitted to the \texttt{arxiv.org}, \texttt{biorxiv.org} and \texttt{medrxiv.org} preprint servers.
\begin{figure*}[!t]
    \centering
    \includegraphics[width=0.49\linewidth]{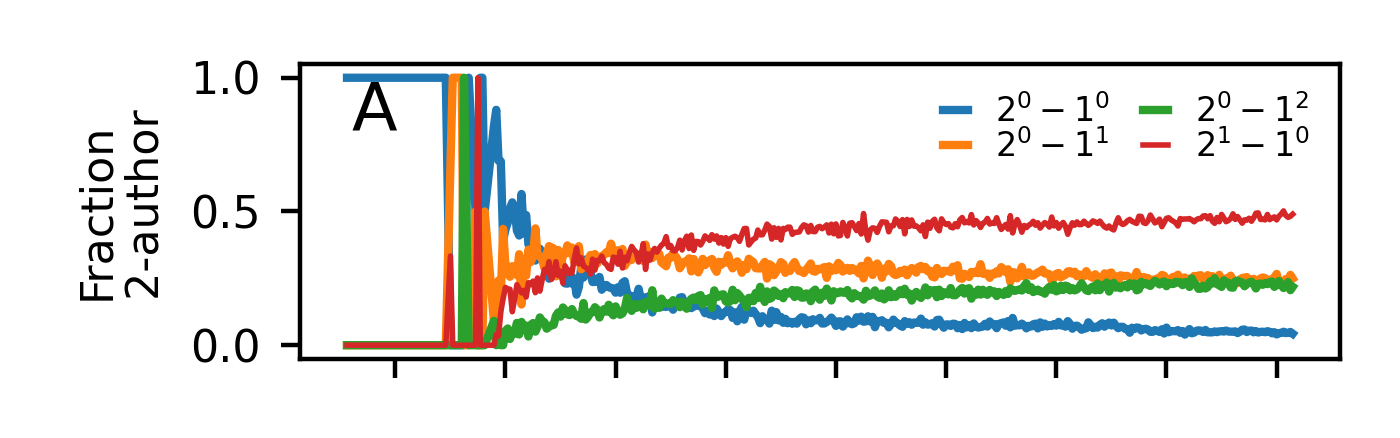}
    \hfill
    \includegraphics[width=0.49\linewidth]{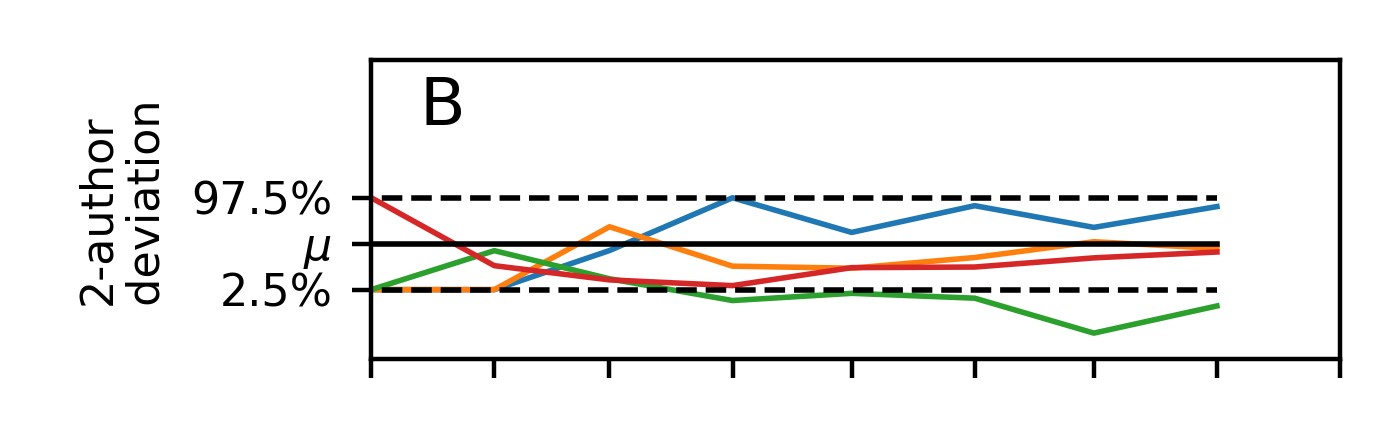}    
\\
\vspace{-10pt}
    \includegraphics[width=0.49\linewidth]{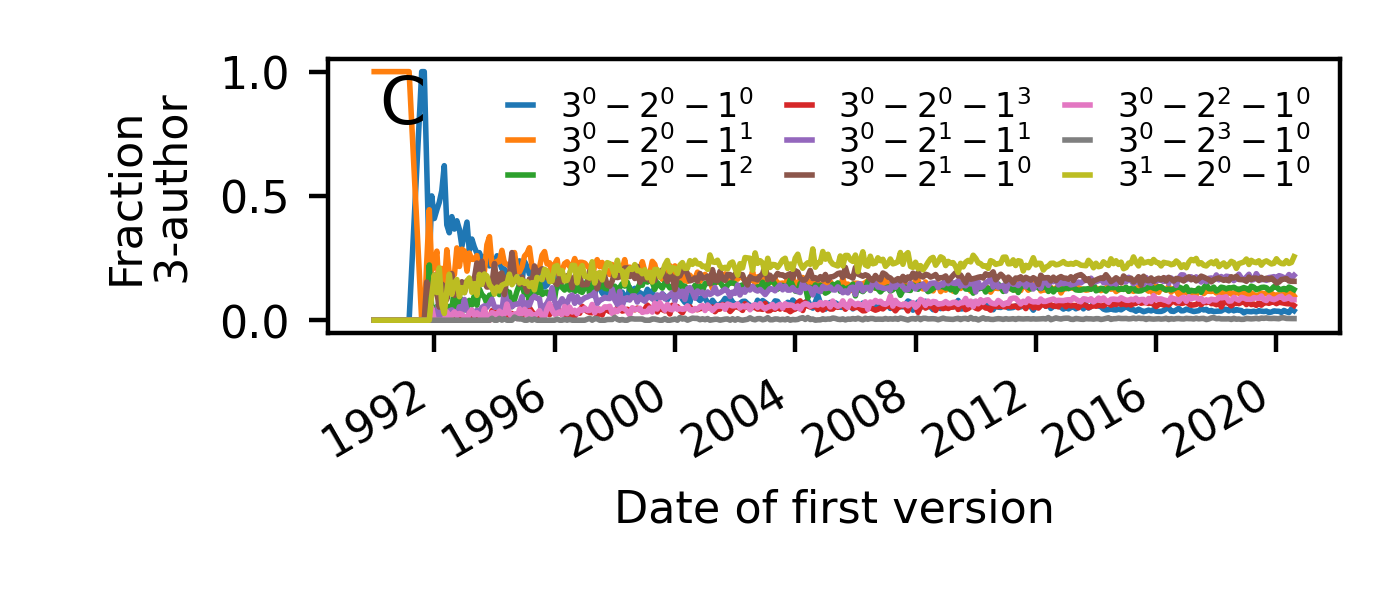}
    \hfill
    \includegraphics[width=0.49\linewidth]{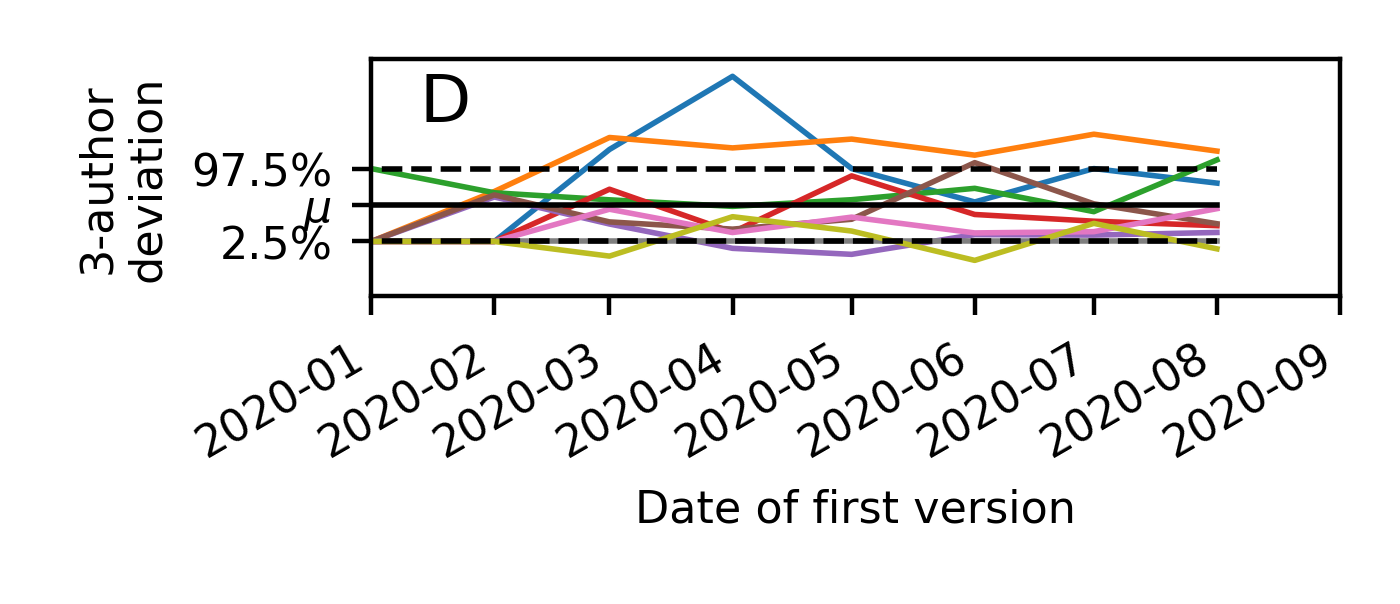}\\
    \vspace{-15pt}
    \caption{
    \textbf{A} Frequency of $2$-patterns for new collaborations in scientific preprints (on \texttt{arXiv.org}, \texttt{biorxiv.org} and \texttt{medrxiv.org}) as a function of time. 
    \textbf{B} Illustration of deviation of $m$-pattern frequency among collaborating scientists on early COVID-19 papers as compared to expectation from the general body of preprints. $\mu$ indicates expectation and dashed lines the $2.5$ and $97.5$ percentiles. 
    \textbf{C-D} As in (\textbf{A}) and (\textbf{B}) but for $3$-patterns.
    }
    \label{fig:COVID}
\end{figure*}

Figure~\ref{fig:COVID}A,C show what fraction of authors on new $2$-author and $3$-author papers had prior relationships that could be summarized by different $m$-patterns. 
The curves are shown as a function of time; time running from the first datapoints for \texttt{arxiv.org} and until September 1, 2020. Datapoints are averages of all papers uploaded in a given month. As the coauthorship hypergraph grows, the likelihood of different prior relationship structures leading to a new $m$-author paper changes. We speculate that each of these curves converges to some value with time. For both $3$ and $2$-author preprints, the repeat collaboration is the most frequent collaboration structure in $2020$.

During the spring of $2020$, 
a surge of COVID-19 related papers accompanied the rising pandemic. 
Teams working on early COVID-19 papers must have formed quickly, and worked intensively to analyze the disease and its consequences. Keeping the common collaboration structures found in Figure~\ref{fig:COVID}A,C in mind, one might wonder whether collaboration structures looked different for these papers that were induced by the external shock of the pandemic. For example, related previous work has established that for a particular subset of these papers \--- multidisciplinary COVID-19 papers \--- collaborations were smaller and more diverse than other collaborations~\cite{cunningham2021collaboration}.

Figure~\ref{fig:COVID}B,D compare the collaboration structure of COVID-19 papers in our dataset to the collaboration structure frequencies found in the entire dataset (COVID-19 papers defined as papers with at least one of the following words in the abstract: \texttt{covid}, \texttt{covid19}, \texttt{covid-19}, \texttt{sars-cov-2}, \texttt{sars-cov2}). If $n_i$ COVID-19 papers were uploaded in month $i$, we compare the frequency of the pattern $\mpat$ to how often we would obtain that pattern when drawing $n_i$ preprints uniformly randomly from all preprints in month $i$. In the data there are significant differences in collaboration structure of COVID-19 papers released between January and August 2020 as compared to papers on all topics in the same period. 
For $2$-author papers, collaborations between two scientists with prior publications but no past joint papers happen less than expected. For $3$-author papers, we find more collaborations consisting of two newcomers and a scientist with prior publications than expected.

\section{Relation between team structure and citation count}

A question that has attracted considerable attention in the literature, is whether team structure influences team performance \cite{nielsen2011reinventing,guimera_team_2005,wu_large_2019,troster_structuring_2014, zeng2021fresh}. Previous studies have examined correlations between performance of teams and team size or dyadic team network structure. Here, we investigate the relation between higher-order team structure---in the form of $m$-patterns---in scientific collaborations and team performance (crudely estimated as the number of citations of published work). 

We study scientific collaborations and their success using the Open Academic Graph (MAG) data set (version 1)~\cite{tang2008arnetminer,sinha2015overview}. The dataset contains more than $166$ million papers including information such as author names, affiliations, publication year, number of citations at the time of data collection, field of study (in the form of keywords) and more. 

To assess whether team structure might affect team performance, it is necessary to consider a number of other variables that could influence how many citations a publication receives. For example, citations could depend on the field of study, the age of the paper, whether the authors on the publication publish in the field often or rarely, and whether they generally receive many citations on their publications. 

To control for the factors other than team structure that could influence citation count, we analyze the data as follows. First, we only compare citation counts for papers within the same field of study. We examine papers from $4$ fields of study: Computer Science, Geology, Mathematics and Sociology. We gather papers from each field of study in separate data sets including only papers where the field of interest is a keyword in the paper's MAG ``field of study'' data. For each of the $4$ fields, we construct an academic collaboration network from the gathered papers and determine the $m$-pattern collaboration structure of each paper. Second, to resolve whether citations are correlated to team structure or other variables, we use a linear regression model to predict the number of citations of a paper based on other variables that could influence citation count: Paper age, mean number of citations of paper authors, mean number of publications of paper authors, and the mean time since paper authors published their first paper. We train the model on $80\%$ of a  dataset that is balanced such that it contains equally many papers with the team structures under consideration (we focus on $2$ kinds of team structures: repeat collaborations and first-time collaborations with no first-time authors), and such that these two sets of papers have identical age distributions (for two sets of papers $A$ and $B$, each with $A(y)$ and $B(y)$ papers of age $y$, we create two subsampled datasets with identical age distributions, $\tilde{A}$ and $\tilde{B}$; these include $min(A(y),B(y))$ published in year $y$ from $A$ and $B$, respectively, drawing papers uniformly at random without replacement from the original sets). For the remaining $20\%$ of papers, we compute the deviation between citations as predicted by the model and actual citations. We quantify this deviation as a mean fractional error of the citation prediction $x_{\rm predicted}$ to the actual citation number $x_{\rm actual}$, 
\[
\mu_i=
\frac{x_{{\rm actual},i}-x_{{\rm predicted},i}}{x_{{\rm actual},i}},\] 
where we set $i=1$ for first-time collaborations and $i=2$ for repeat collaborations.
Finally, we evaluate to what degree the model underestimated citation count of repeat collaborations compared to first-time collaborations or vice versa by performing two-sample tests for these summary statistics. 

Tabel~\ref{tab:citation_count} shows our results for 2-author and 3-author papers. For Computer Science, 2-author repeat collaborations get more citations than would be expected from the trained model alone; moreover, the 2-author repeat collaborations outperform model expectations to a statistical significant higher degree than is done by 2-author first-time collaborations. For 3-author Computer Science collaborations the result is the opposite: 3-author first-time collaborations outperform the model prediction to a statistically significant level compared to the repeat collaboration. 
For Geology, the findings for 2-author papers mimic those found for Computer Science: the 2-author repeat collaborations outperform model expectations to a statistical significant higher degree than is done by 2-author first-time collaborations. For 3-author Geology collaborations there is no significant difference between deviation from the model prediction for repeat and first-time collaborations. 
For Mathematics, there is a significant difference between deviation from the model prediction for repeat and first-time collaborations for 3-author papers. In this case, the finding is similar to that found for 3-author Computer Science papers: first-time collaborations outperform the model prediction to a statistically significant level compared to the repeat collaborations. For Sociology collaborations with 2 or 3 authors, there is no significant difference between deviation from the model prediction for repeat and first-time collaborations. 

\begin{table*}[]
    \centering
    \begin{tabular}{|ccccc|}
    \hline
    \hline
         \textbf{Field} & \textbf{Team size}  & $\boldsymbol{\mu_1\pm\sigma_{\mu_1}}$ & $\boldsymbol{\mu_2\pm\sigma_{\mu_2}}$ & $\boldsymbol{z}$\textbf{-score}  \\
         \hline 
    Computer Science & \makecell{2 authors \\ 3 authors}  &\makecell{$0.3521 \pm 0.0043$ \\ $0.3818\pm0.0055$}  & \makecell{$0.3690\pm 0.0046$\\$0.3554\pm 0.0058$} & \makecell{$2.670$\\$3.304$} \\         
    \hline
    Geology & \makecell{2 authors \\ 3 authors}   &\makecell{$0.2342 \pm 0.0055$ \\ $0.1377\pm0.0050$}  & \makecell{$0.2571\pm0.0064$\\$0.1369\pm0.0049$} & \makecell{2.717\\0.102} \\
    \hline
    Mathematics & \makecell{2 authors \\ 3 authors}  &\makecell{$0.2155 \pm 0.0039$ \\ $0.2846\pm 0.0052 $}  & \makecell{$0.2199\pm 0.0042$\\$0.2697 \pm 0.0050$} & \makecell{0.762\\2.055} \\
    \hline
    Sociology & \makecell{2 authors \\ 3 authors}  &\makecell{$0.3184 \pm 0.0179$ \\ $0.2164\pm 0.0178$}  & \makecell{$0.2908\pm 0.0145$\\$0.2016\pm 0.0179$} & \makecell{$1.200$\\$0.585$} \\ \hline   
    \hline    
    \end{tabular}
    \caption{Relationship between past collaborations and future citations of academic papers in $4$ different scientific fields. For each field, we make two comparisons. In the first, we compare citations of  2-author papers where both authors have published in the past, but never together (the $2$-pattern $2^0-1^2$), to citations of 2-author papers where the authors have collaborated with each other in the past (the $2$-pattern $2^1-1^0$). In the second comparison, we compare citations of $3$-author papers where all authors have published in the past but never all three together (the union of the $3$-patterns $3^0-2^0-1^3$, $3^0-2^1-1^1$, $3^0-2^2-1^0$ and $3^0-2^3-1^0$) to citations of 3-author papers where the triple of authors is a repeat collaboration (the $3$-pattern $3^1-2^0-1^0$). 
    We train a linear regression to predict citation count from the 4 properties of a paper: paper age, mean number of past publications of the authors, mean number of citations of the authors, and mean time since the authors published their first papers. 
    $\mu_1$ indicates the mean error (see main text) on citation predictions for first-time collaborations and $\mu_2$ the mean error on predictions for repeat collaborations. Lastly, we estimate whether these mean fractional errors are significantly different from each other by computing the $z$-score of the pairs of estimates.
    }
    \label{tab:citation_count}
\end{table*}

\section{Discussion}
In many different contexts, individual nodes occasionally co-occur together. Understanding which nodes are likely to form such collaborations, and how prior relationships influence collaboration outcome is important to study. In this work, we have introduced the concept of $m$-patterns, a new family of structural patterns that quantify prior relationships between $m$ nodes in a hypergraph. 

We have argued that prevalence of different such $m$-patterns should depend on hypergraph characteristics such as density of hyperedges, and we have quantified these expectations by studying a $G^{(m)}(N,p)$ model. In particular, we have derived analytical expressions for $m$-pattern prevalence and provided proofs that 
some patterns are and others can never be extreme in the $G^{(m)}(N,p)$ model 
in the limit $N\to \infty$. 

Comparing the model to data from different domains, we found both similarities and differences. Most strikingly, we found that some datasets had certain patterns overrepresented by several orders of magnitude as compared to the model expectation. Interestingly, datasets from the same domain often had similar discrepancies as compared to the model.

In the dataset of preprints, we found the repeat collaboration to be the most prevalent for both $2$-author and $3$-author papers. This is interesting because such a finding would only take place in very dense networks if collaborations were happen uniformly randomly. 
We proceeded to examine whether collaboration structure was different for early COVID-19 preprints as compared to the full dataset of preprints. We found that $2$-author papers were less often coauthored by two scientists with prior publications but no collaborations. For $3$-author preprints, we found more collaborations structures consisting of two newcomers and a person with previous publications.

Finally, we examined whether team structure of academic papers correlated with future citation counts. Considering 2-author and 3-author publications separately, we compared citations of first-time collaboration without first-time authors to citations of repeat collaborations. We did so for 4 fields of study: Computer Science, Geology, Mathematics and Sociology. To account for other factors that could influence future citation count, we trained a linear regression model to predict future citation counts based on paper and author count. In some cases, differences in citation predictions and actual citations for first-time and repeat collaborations were statistically significant. For Computer Science and Geology, 2-author repeat collaborations outperformed model expectations to a statistically significantly higher degree than was done by 2-author first-time collaborations. For Computer Science and Mathematics, 3-author first-time collaborations outperformed model expectations to a statistically significantly higher degree than was done by 2-author repeat collaborations. The linear model is crude and for all fields it tended to underestimate citation count by between $13\%$ and $39\% $ of the actual future citation counts. This being said, the consistency of the results speak to their trustworthiness: We found that repeat collaborations had better performance for 2-author collaborations whereas first-time collaborations had better performance for 3-author collaborations.

There are several natural future research directions related to our work. Throughout this paper, we have argued that investigating whether team structure correlates with team performance is an interesting question. Although we did examine this for 2-author and 3-author papers from $4$ fields, there are many promising questions in this direction. We found different results for 2-author and 3-author papers; what happens for larger collaborations? And if repeat collaborations tend to have higher or lower performance, is the effect larger, smaller or unchanged for teams that collaborate over and over again? Our investigation of whether datasets from the same domains tend to have the same $m$-patterns over and underrepresented remains qualitative. An obvious next step would be to attempt to train an algorithm to guess the domain that a hypergraph stems from given only information about $m$-pattern prevalence. We note that such investigations should carefully control for the fact that data from different domains typically cover different orders of magnitudes of the hyperedge density $p$. Finally, we note that collaboration hypergraphs such as the preprint coauthorship network are growing systems. Although models for collaboration networks exist~\cite{guimera_team_2005}, these are based on dyadic interactions. Formulating a growth model that gives rise to correct $m$-pattern frequencies is an open question.
\clearpage
\bibliographystyle{siamplain}
\bibliography{m-patterns_bib.bib}
\newpage
\appendix

\section{Illustration of hypergraph definitions}
\label{app:sec:hypergraph_definitions_illustration}

In Figure~\ref{app:fig:hypergraph_concepts}, we illustrate some of the hypergraph concepts introduced in Section~\ref{sec:definition}.
\begin{figure}[b]
    \centering
    \includegraphics[width=0.35\linewidth]{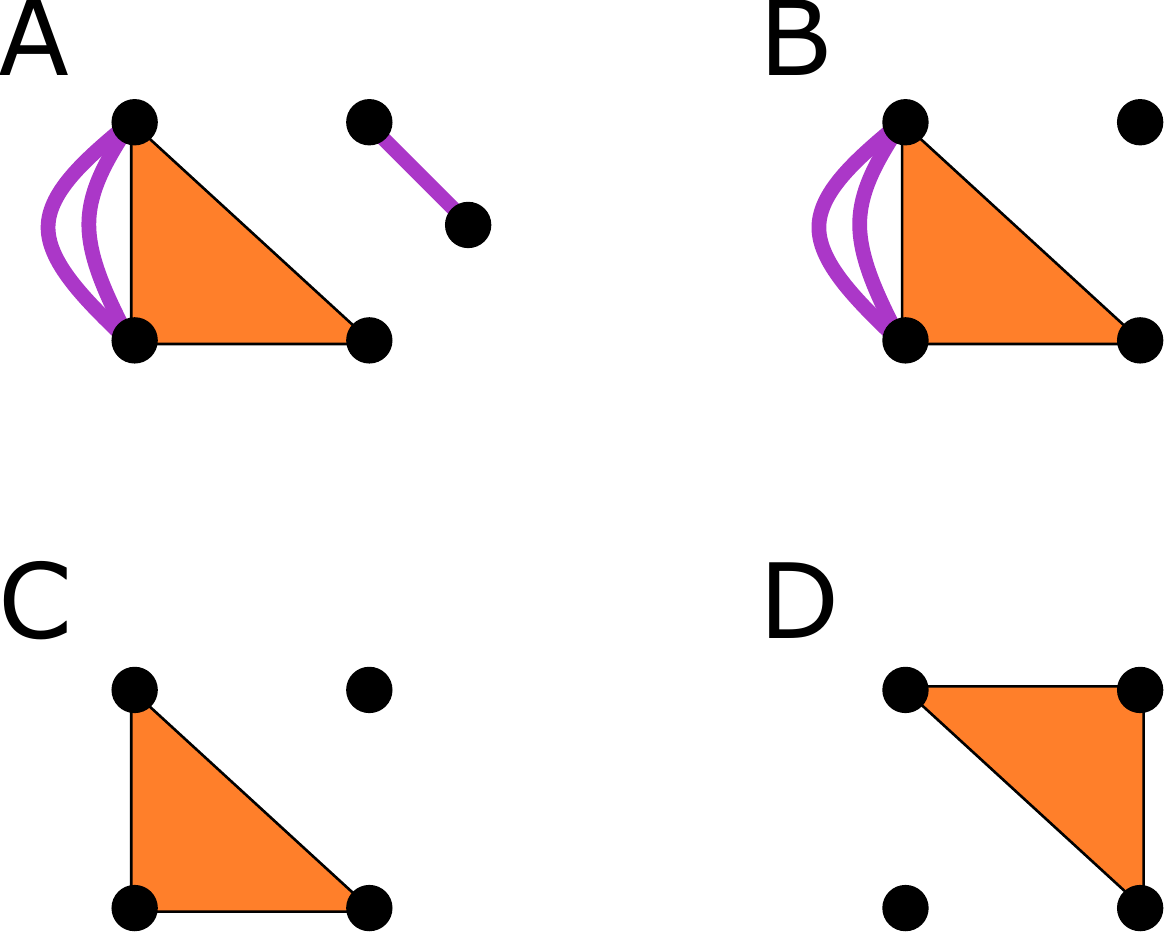}
    \caption{
    Illustration of some concepts introduced in Section~\ref{sec:definition}. \textbf{A} $5$-node hypergraph. \textbf{B} Induced subhypergraph of (A) on the $4$ left-most nodes. \textbf{(C)} Maximal induced subhypergraph of (A) on $4$ left-most nodes. \textbf{D} A $4$-pattern. (C) happens to be an instance of this $4$-pattern in the hypergraph in (A). Labelling nodes $\{0,1,2,3\}$ starting with the label $0$ in  the top-left corner and increasing labels by $1$ in the clockwise direction, (C) and (D) are also examples of two different labelled $4$-patterns.
    }
 \label{app:fig:hypergraph_concepts}
\end{figure}

\section{Labelled version of Figure~\ref{fig:random_m-patterns}B}
\label{app:sec:labelled2B}
In Figure~\ref{fig:random_m-patterns}B we did label individual curves. In Figure~\ref{app:fig:labelled2B} we split the curves from Figure~\ref{fig:random_m-patterns}B into $3$ different panels and provide curve labels. We use a different naming convention here than we did for $3$-patterns in the rest of the manuscript. We do so because the other (less complicated) notation could uniquely describe $1$- $2$- and $3$-patterns, but cannot uniquely map all $4$-patterns. The naming convention is described in the caption of Figure~\ref{app:fig:labelled2B}.
\begin{figure}[t]
    \centering
    \includegraphics{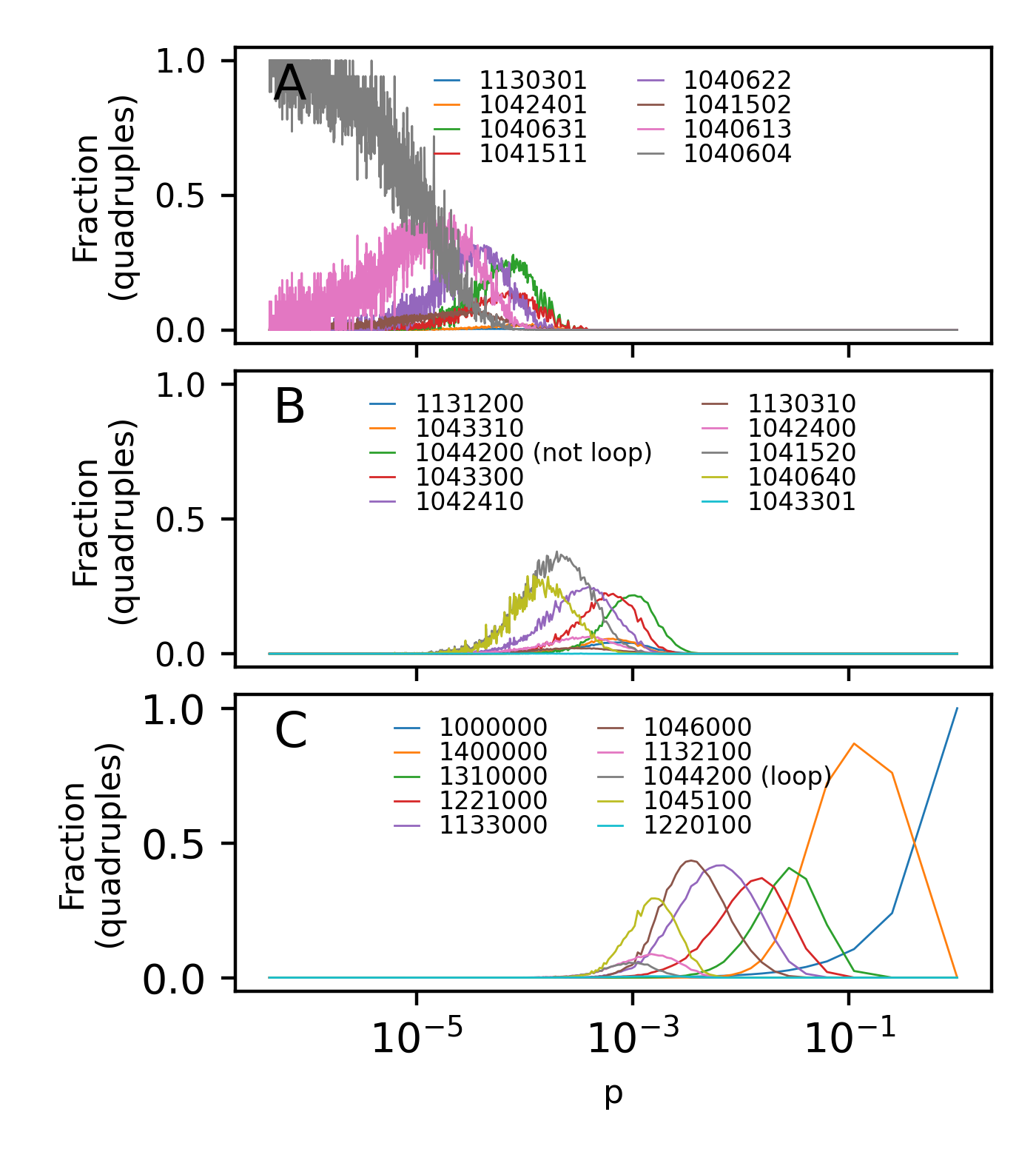}
    \caption{
    Labelled version of Figure~\ref{fig:random_m-patterns}B split into $3$ panels to make plot colors easier to distinguish. The naming convention used for $4$-patterns is different than that used for $3$-patterns in the rest of the manuscript. The pattern $1ABCDEF$ has A $3$-node hyperedges filled and B not filled, C $ 2$-node hyperedges filled and D not filled, E $1$-node hyperedges filled and F not filled. The pattern $1000000$ is the $4$-pattern consisting of a single $4$-node hyperedge. There are $2$ possible $4$-patterns with the name $1044200$ (consisting of 4 $2$-node hyperedges and all other possible hyperedges missing): one where the hyperedges form a loop and another where they do not. The analytical solutions are not plotted in this figure. 
    }
    \label{app:fig:labelled2B}
\end{figure}
\section{Proof of Theorem~\ref{thm:pure_patterns}}
\label{app:proof:thm:pure_patterns}
\begin{proof}
Without loss of generality, let us focus on the pure pattern consisting only of $k$-node hyperedges. Let us denote the prevalence of this pattern $P(\mpat_k)$. By Eq.~\eqref{lem:mpat_analytical}, the analytical formula for $P(\mpat_k)$ is,
\begin{equation}
    P(\mpat_k)= p_k^{{m\choose k}}\prod_{i=k+1}^m (1-p_i)^{{m\choose i}}.
\end{equation}
According to Lemma~\ref{lem:zeroone}, for any $\epsilon>0$ and large enough $N$, we can choose a $p$ such that all factors in this product are arbitrarily large. All other patterns will either have factors of $(1-p_k)$ in the analytical expression, or factors of $p_l$ in the expression, where $l\ge k+1$. By Lemma~\ref{lem:zeroone}, $N$ can be chosen large enough to make any such factors arbitrarily close to $0$ if $0<p_k<1$. This proves the theorem.
\end{proof}

\section{Proof of Lemma~\ref{lem:pure_dominates_above_1/2}}
\label{app:proof:lem:pure_dominates_above_1/2}
\begin{proof}
By Lemma~\ref{lem:zeroone}, if $p_k>\frac{1}{2}$, $N$ can be chosen large enough to make the value of $p_l$ arbitrarily close to $1$ for $l\le k-1$. 
The prevalence of the pure pattern with $k$-node hyperedges is 
\begin{equation}
    P(\mpat_k) = p_k^{{m \choose k}}\Omega,
    \label{eq:kpure}
\end{equation}
where $\Omega=\prod_{i=k+1}^m (1-p_i)^{{m\choose i}}$. 
The prevalence of the non-pure patterns containing $x_k$ $k$-node and $x_{k-1}$ $(k-1)$-node hyperedges is,
\begin{equation}
    P(\mpat_{k-1,k}') = p_{k-1}^{x_{k-1}} p_k^{x_k} (1-p_k)^{{m\choose k}-x_k} \Omega.
    \label{eq:not_kpure}
\end{equation}
By Lemma~\ref{lem:zeroone}, for any $\epsilon>0$ and large enough $N$, the first factor in this last expression can get arbitrarily close to $1$. 
In this limit, the pure and non-pure patterns therefore cross when the remaining factors in Eqs.~\eqref{eq:kpure} and \eqref{eq:not_kpure} are equal. This happens at when $p_k=(1-p_k)$; in other words, $p_k = \frac{1}{2}$. 
For any $p_k>\frac{1}{2}$, $p_k > (1-p_k)$. Comparing Eqs.~\eqref{eq:kpure} and \eqref{eq:not_kpure}, the pure pattern dominates in this case.
\end{proof}

\section{Proof of 
Lemma \ref{lem:pure_dominates_below_1/2}}
\label{app:proof:lem:pure_dominates_below_1/2}
\begin{proof}
By Lemma~\ref{lem:zeroone}, if $p_{k+1}$ is non-zero, we can choose $N$ large enough $p_l\to1$ as $N\to \infty$ for $l\le k-1$. 
The prevalence of the pure pattern with $k$-node hyperedges is 
\begin{equation}
    P(\mpat_k) = p_k^{{m \choose k}}(1-p_{k+1})^{{m\choose k+1}}\Omega',
    \label{eq:kpure_elab}    
\end{equation}
where $\Omega'=\prod_{i=k+2}^m (1-p_i)^{{m\choose i}}$. 
The prevalence of the non-pure patterns containing $x_k$ $k$-node and $x_{k+1}$ $(k+1)$-node hyperedges is,
\begin{equation}
    P(\mpat_{k-1,k}') = p_{k}^{x_{k}} p_{k+1}^{x_{k+1}} (1-p_{k+1})^{{m\choose k+1}-x_{k+1}} \Omega'.
    \label{eq:k+1_nonpure}
\end{equation}
By Lemma~\ref{lem:zeroone}, for any $\epsilon>0$ and large enough $N$, the first factor in both Eqs~\eqref{eq:kpure_elab} and \eqref{eq:k+1_nonpure} can get arbitrarily close to $1$. 
Hence, the pure and non-pure patterns above cross when the remaining factors are equal in Eqs~\eqref{eq:kpure_elab} and \eqref{eq:k+1_nonpure}. This happens when $p_{k+1}=\frac{1}{2}$. 
For lower values of $p_{k+1}$, $(1-p_{k+1})>p_{k+1}$. 
This proves the lemma.
\end{proof}

\section{Illustration from proof of Theorem~\ref{thm:non-pure}}
\label{app:illustration:thm:non-pure}
In the proof of Theorem~\ref{thm:non-pure}, we apply  Lemma~\ref{lem:cross_1/2}. This Lemma states that many prevalence curves cross at values of $p$ where $p_k=1/2$. In Figure~\ref{fig:labelled}, we confirm this by simulations of the $G^{(m)}(N,p)$ model.
\begin{figure}
    \centering
    \includegraphics[]{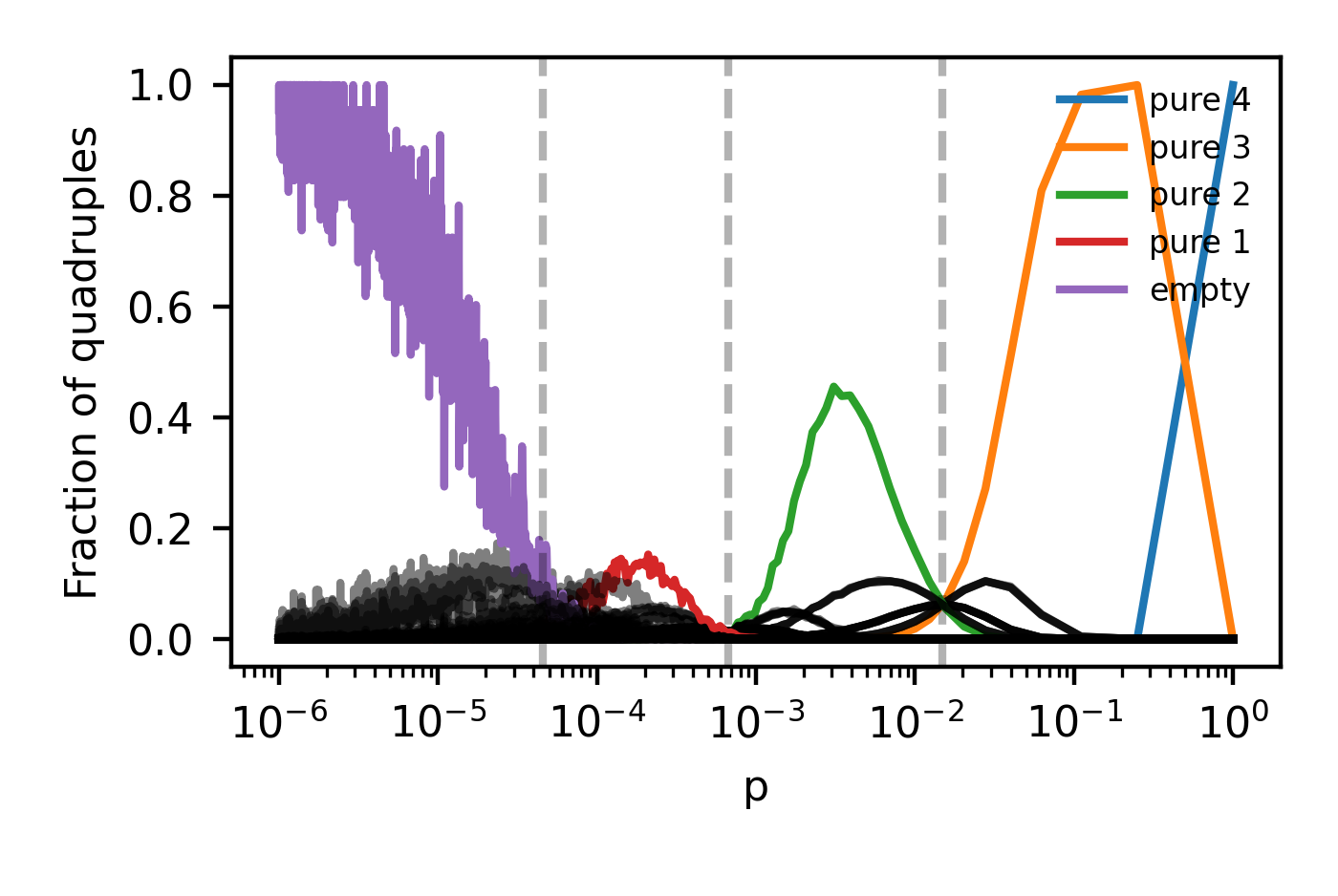}
    \caption{Frequency of labelled $m$-patterns in the $G^{(m)}(N,p)$ model with pure $k$-node hyperedge patterns in colors. Each datapoint plots the average prevalence of a labelled $m$-pattern in $10$ simulations of the model for the given $p$ value and $m=4$, $N=100$. Vertical gray dashed lines indicate values pf $p$ where $p_k=1/2$ for $1\le k\le4$. As argued in the proof of Theorem~\eqref{thm:non-pure}, many prevalence curves cross at these values of $p$.}
    \label{fig:labelled}
\end{figure}
\section{Proof of Theorem~\ref{thm:two_patterns_not_extreme}}
\label{app:proof:thm:two_patterns_not_extreme}
\begin{proof}
Let us refer to the patterns as $\mpat_A$ and $\mpat_B$. 
The prevalence of the patterns can be written down explicitly,
\begin{align}
    P(\mpat_A) &= \gamma_A p_{k-1}^{x_{k-1}^{(A)}}p_{k}^{x_{k}}(1-p_k)^{{m\choose k}-x_k}\Omega,\\
    P(\mpat_B) &= \gamma_B p_{k-1}^{x_{k-1}^{(B)}}p_{k}^{x_{k}}(1-p_k)^{{m\choose k}-x_k}\Omega,
\end{align}
where $\Omega=\prod_{i=k+1}^m (1-p_i)^{{m\choose i}}$ and $x_j^{(L)}$ is the number of $j$-node hyperedges in $\mpat_L$. 
By Lemma~\ref{lem:zeroone}, $N$ can be chosen large enough to make the $p_{k-1}$ factor arbitrarily close to $1$ for both $P(\mpat_A)$ and $P(\mpat_B)$. 
Hence, for increasing $N$,
\begin{equation}
    \frac{P(\mpat_A)}{P(\mpat_B)} \to \frac{\gamma_A}{\gamma_B}.
\end{equation}
We conclude that the pattern with the smallest combinatorial factor is bound to be less prevalent than the other pattern.
\end{proof}
\section{Code and Data}
Code and data to reproduce the plots in this paper are available in the Github directory: \texttt{https://github.com/jonassjuul/team-formation}.
\end{document}

%% file: ex_shared.tex
\usepackage{lipsum}
\usepackage{amsfonts}
\usepackage{graphicx}
\usepackage{epstopdf}
\usepackage{algorithmic}
\ifpdf
  \DeclareGraphicsExtensions{.eps,.pdf,.png,.jpg}
\else
  \DeclareGraphicsExtensions{.eps}
\fi

\usepackage{enumitem}
\setlist[enumerate]{leftmargin=.5in}
\setlist[itemize]{leftmargin=.5in}

\newsiamremark{remark}{Remark}
\newsiamremark{hypothesis}{Hypothesis}
\crefname{hypothesis}{Hypothesis}{Hypotheses}
\newsiamthm{claim}{Claim}

\headers{Hypergraph patterns and collaboration structure}{J. L. Juul, A. R. Benson, and J. Kleinberg}

\title{Hypergraph patterns and collaboration structure
}

\author{Jonas L. Juul \thanks{Center for Applied Mathematics, Cornell University, Ithaca, New York 14853, USA (\email{jjuul@cornell.edu})}
\and Austin R. Benson\thanks{Department of Computer Science, Cornell University, Ithaca, New York 14853, USA.}
\and Jon Kleinberg\footnotemark[2]}

\usepackage{amsopn}

\makeatletter
\newcommand*{\addFileDependency}[1]{
  \typeout{(#1)}
  \@addtofilelist{#1}
  \IfFileExists{#1}{}{\typeout{No file #1.}}
}
\makeatother



%% file: main_mpatterns.bbl
\begin{thebibliography}{10}

\bibitem{ahn_hypergraph_2018}
{\sc K.~Ahn, K.~Lee, and C.~Suh}, {\em Hypergraph {Spectral} {Clustering} in
  the {Weighted} {Stochastic} {Block} {Model}}, IEEE Journal of Selected Topics
  in Signal Processing, 12 (2018), pp.~959--974,
  \url{https://doi.org/10.1109/JSTSP.2018.2837638},
  \url{https://ieeexplore.ieee.org/document/8360467/} (accessed 2021-06-15).

\bibitem{battiston2021physics}
{\sc F.~Battiston, E.~Amico, A.~Barrat, G.~Bianconi, G.~Ferraz~de Arruda,
  B.~Franceschiello, I.~Iacopini, S.~K{\'e}fi, V.~Latora, Y.~Moreno, et~al.},
  {\em The physics of higher-order interactions in complex systems}, Nature
  Physics, 17 (2021), pp.~1093--1098.

\bibitem{battiston_networks_2020}
{\sc F.~Battiston, G.~Cencetti, I.~Iacopini, V.~Latora, M.~Lucas, A.~Patania,
  J.-G. Young, and G.~Petri}, {\em Networks beyond pairwise interactions:
  {Structure} and dynamics}, Physics Reports, 874 (2020), pp.~1--92,
  \url{https://doi.org/10.1016/j.physrep.2020.05.004},
  \url{https://linkinghub.elsevier.com/retrieve/pii/S0370157320302489}
  (accessed 2021-06-15).

\bibitem{benson_simplicial_2018}
{\sc A.~R. Benson, R.~Abebe, M.~T. Schaub, A.~Jadbabaie, and J.~Kleinberg},
  {\em Simplicial closure and higher-order link prediction}, Proceedings of the
  National Academy of Sciences, 115 (2018), pp.~E11221--E11230,
  \url{https://doi.org/10.1073/pnas.1800683115},
  \url{http://www.pnas.org/lookup/doi/10.1073/pnas.1800683115} (accessed
  2021-06-15).

\bibitem{benson_tensor_2015}
{\sc A.~R. Benson, D.~F. Gleich, and J.~Leskovec}, {\em Tensor {Spectral}
  {Clustering} for {Partitioning} {Higher}-order {Network} {Structures}}, in
  Proceedings of the 2015 {SIAM} {International} {Conference} on {Data}
  {Mining} ({SDM}), Proceedings, Society for Industrial and Applied
  Mathematics, June 2015, pp.~118--126,
  \url{https://doi.org/10.1137/1.9781611974010.14},
  \url{https://epubs.siam.org/doi/10.1137/1.9781611974010.14} (accessed
  2021-06-15).

\bibitem{benson_higher-order_2016}
{\sc A.~R. Benson, D.~F. Gleich, and J.~Leskovec}, {\em Higher-order
  organization of complex networks}, Science, 353 (2016), pp.~163--166,
  \url{https://doi.org/10.1126/science.aad9029},
  \url{https://www.sciencemag.org/lookup/doi/10.1126/science.aad9029} (accessed
  2021-06-15).

\bibitem{benson_higher-order_2021}
{\sc A.~R. Benson, D.~R. Gleich, and D.~J. Higham}, {\em Higher-order {Network}
  {Analysis} {Takes} {Off}, {Fueled} by {Old} {Ideas} and {New} {Data}}, Jan.
  2021,
  \url{https://sinews.siam.org/Details-Page/higher-order-network-analysis-takes-off-fueled-by-old-ideas-and-new-data}
  (accessed 2021-06-15).

\bibitem{berge1984hypergraphs}
{\sc C.~Berge}, {\em Hypergraphs: combinatorics of finite sets}, vol.~45,
  Elsevier, Amsterdam, Netherlands, 1989.

\bibitem{bianconi2017emergent}
{\sc G.~Bianconi and C.~Rahmede}, {\em Emergent hyperbolic network geometry},
  Scientific reports, 7 (2017), pp.~1--9.

\bibitem{bretto2013hypergraph}
{\sc A.~Bretto}, {\em Hypergraph theory}, An introduction. Mathematical
  Engineering. Cham: Springer,  (2013).

\bibitem{chien_community_2018}
{\sc I.~Chien, C.-Y. Lin, and I.-H. Wang}, {\em Community {Detection} in
  {Hypergraphs}: {Optimal} {Statistical} {Limit} and {Efficient} {Algorithms}},
  in International {Conference} on {Artificial} {Intelligence} and
  {Statistics}, PMLR, Mar. 2018, pp.~871--879,
  \url{http://proceedings.mlr.press/v84/chien18a.html} (accessed 2021-06-15).

\bibitem{chodrow_configuration_2020}
{\sc P.~S. Chodrow}, {\em Configuration models of random hypergraphs}, Journal
  of Complex Networks, 8 (2020), \url{https://doi.org/10.1093/comnet/cnaa018},
  \url{https://doi.org/10.1093/comnet/cnaa018} (accessed 2021-06-15).

\bibitem{chodrow_generative_2021}
{\sc P.~S. Chodrow, N.~Veldt, and A.~R. Benson}, {\em Generative hypergraph
  clustering: from blockmodels to modularity}, arXiv:2101.09611 [physics,
  stat],  (2021), \url{http://arxiv.org/abs/2101.09611} (accessed 2021-06-15).
\newblock arXiv: 2101.09611.

\bibitem{costa2016random}
{\sc A.~Costa and M.~Farber}, {\em Random simplicial complexes}, in
  Configuration spaces, Springer, 2016, pp.~129--153.

\bibitem{courtney2016generalized}
{\sc O.~T. Courtney and G.~Bianconi}, {\em Generalized network structures: The
  configuration model and the canonical ensemble of simplicial complexes},
  Physical Review E, 93 (2016), p.~062311.

\bibitem{courtney2017weighted}
{\sc O.~T. Courtney and G.~Bianconi}, {\em Weighted growing simplicial
  complexes}, Physical Review E, 95 (2017), p.~062301.

\bibitem{cunningham2021collaboration}
{\sc E.~Cunningham, B.~Smyth, and D.~Greene}, {\em Collaboration in the time of
  covid: a scientometric analysis of multidisciplinary sars-cov-2 research},
  Humanities and Social Sciences Communications, 8 (2021), pp.~1--8.

\bibitem{dyer_sampling_2020}
{\sc M.~Dyer, C.~Greenhill, P.~Kleer, J.~Ross, and L.~Stougie}, {\em Sampling
  hypergraphs with given degrees}, arXiv:2006.12021 [cs],  (2020),
  \url{http://arxiv.org/abs/2006.12021} (accessed 2021-06-15).
\newblock arXiv: 2006.12021.

\bibitem{easley_networks_2010}
{\sc D.~Easley and J.~Kleinberg}, {\em Networks, crowds, and markets: reasoning
  about a highly connected world}, Cambridge University Press, New York, 2010.
\newblock OCLC: ocn495616815.

\bibitem{guimera_team_2005}
{\sc R.~Guimera, B.~Uzzi, J.~Spiro, and L.~A.~N. Amaral}, {\em Team {Assembly}
  {Mechanisms} {Determine} {Collaboration} {Network} {Structure} and {Team}
  {Performance}}, Science, 308 (2005), pp.~697--702,
  \url{https://doi.org/10.1126/science.1106340},
  \url{https://www.sciencemag.org/lookup/doi/10.1126/science.1106340} (accessed
  2021-06-15).

\bibitem{jackson_social_2008}
{\sc M.~O. Jackson}, {\em Social and economic networks}, Princeton Univ. Press,
  Princeton, NJ, 2008.
\newblock OCLC: 254984264.

\bibitem{kahle2009topology}
{\sc M.~Kahle}, {\em Topology of random clique complexes}, Discrete
  mathematics, 309 (2009), pp.~1658--1671.

\bibitem{kahle2014topology}
{\sc M.~Kahle et~al.}, {\em Topology of random simplicial complexes: a survey},
  AMS Contemp. Math, 620 (2014), pp.~201--222.

\bibitem{kaminski_clustering_2019}
{\sc B.~Kamiński, V.~Poulin, P.~Prałat, P.~Szufel, and F.~Théberge}, {\em
  Clustering via hypergraph modularity}, PLOS ONE, 14 (2019), p.~e0224307,
  \url{https://doi.org/10.1371/journal.pone.0224307},
  \url{https://journals.plos.org/plosone/article?id=10.1371/journal.pone.0224307}
  (accessed 2021-06-15).

\bibitem{kim_stochastic_2018}
{\sc C.~Kim, A.~S. Bandeira, and M.~X. Goemans}, {\em Stochastic {Block}
  {Model} for {Hypergraphs}: {Statistical} limits and a semidefinite
  programming approach}, arXiv:1807.02884 [cs, math],  (2018),
  \url{http://arxiv.org/abs/1807.02884} (accessed 2021-06-15).
\newblock arXiv: 1807.02884.

\bibitem{kook_evolution_2020}
{\sc Y.~Kook, J.~Ko, and K.~Shin}, {\em Evolution of {Real}-world
  {Hypergraphs}: {Patterns} and {Models} without {Oracles}}, arXiv:2008.12729
  [cs],  (2020), \url{http://arxiv.org/abs/2008.12729} (accessed 2021-06-15).
\newblock arXiv: 2008.12729.

\bibitem{kumar_hypergraph_2020}
{\sc T.~Kumar, S.~Vaidyanathan, H.~Ananthapadmanabhan, S.~Parthasarathy, and
  B.~Ravindran}, {\em Hypergraph clustering by iteratively reweighted
  modularity maximization}, Applied Network Science, 5 (2020), p.~52,
  \url{https://doi.org/10.1007/s41109-020-00300-3},
  \url{https://appliednetsci.springeropen.com/articles/10.1007/s41109-020-00300-3}
  (accessed 2021-06-15).

\bibitem{landry2022hypergraph}
{\sc N.~W. Landry and J.~G. Restrepo}, {\em Hypergraph assortativity: A
  dynamical systems perspective}, Chaos: An Interdisciplinary Journal of
  Nonlinear Science, 32 (2022), p.~053113.

\bibitem{lee_how_2021}
{\sc G.~Lee, M.~Choe, and K.~Shin}, {\em How {Do} {Hyperedges} {Overlap} in
  {Real}-{World} {Hypergraphs}? - {Patterns}, {Measures}, and {Generators}}, in
  Proceedings of the {Web} {Conference} 2021, {WWW} '21, Ljubljana, Slovenia,
  Apr. 2021, Association for Computing Machinery, pp.~3396--3407,
  \url{https://doi.org/10.1145/3442381.3450010},
  \url{https://doi.org/10.1145/3442381.3450010} (accessed 2021-06-15).

\bibitem{lee_hypergraph_2020}
{\sc G.~Lee, J.~Ko, and K.~Shin}, {\em Hypergraph {Motifs}: {Concepts},
  {Algorithms}, and {Discoveries}}, Proceedings of the VLDB Endowment, 13
  (2020), pp.~2256--2269, \url{https://doi.org/10.14778/3407790.3407823},
  \url{http://arxiv.org/abs/2003.01853} (accessed 2021-06-15).
\newblock arXiv: 2003.01853.

\bibitem{linial2006homological}
{\sc N.~Linial and R.~Meshulam}, {\em Homological connectivity of random
  2-complexes}, Combinatorica, 26 (2006), pp.~475--487.

\bibitem{meshulam2009homological}
{\sc R.~Meshulam and N.~Wallach}, {\em Homological connectivity of random
  k-dimensional complexes}, Random Structures \& Algorithms, 34 (2009),
  pp.~408--417.

\bibitem{newman_networks_2018}
{\sc M.~E.~J. Newman}, {\em Networks}, Oxford University Press, Oxford, United
  Kingdom ; New York, NY, United States of America, second edition~ed., 2018.

\bibitem{nielsen2011reinventing}
{\sc M.~Nielsen}, {\em Reinventing Discovery: {The} new era of networked
  science}, Princeton University Press, 2013.

\bibitem{sinha2015overview}
{\sc A.~Sinha, Z.~Shen, Y.~Song, H.~Ma, D.~Eide, B.-J. Hsu, and K.~Wang}, {\em
  An overview of microsoft academic service (mas) and applications}, in
  Proceedings of the 24th international conference on world wide web, 2015,
  pp.~243--246.

\bibitem{tang2008arnetminer}
{\sc J.~Tang, J.~Zhang, L.~Yao, J.~Li, L.~Zhang, and Z.~Su}, {\em Arnetminer:
  extraction and mining of academic social networks}, in Proceedings of the
  14th ACM SIGKDD international conference on Knowledge discovery and data
  mining, 2008, pp.~990--998.

\bibitem{troster_structuring_2014}
{\sc C.~Tröster, A.~Mehra, and D.~van Knippenberg}, {\em Structuring for team
  success: {The} interactive effects of network structure and cultural
  diversity on team potency and performance}, Organizational Behavior and Human
  Decision Processes, 124 (2014), pp.~245--255,
  \url{https://doi.org/10.1016/j.obhdp.2014.04.003},
  \url{https://linkinghub.elsevier.com/retrieve/pii/S0749597814000326}
  (accessed 2021-06-15).

\bibitem{veldt_higher-order_2021}
{\sc N.~Veldt, A.~R. Benson, and J.~Kleinberg}, {\em Higher-order {Homophily}
  is {Combinatorially} {Impossible}}, arXiv:2103.11818 [cs],  (2021),
  \url{http://arxiv.org/abs/2103.11818} (accessed 2021-06-15).
\newblock arXiv: 2103.11818.

\bibitem{wu_large_2019}
{\sc L.~Wu, D.~Wang, and J.~A. Evans}, {\em Large teams develop and small teams
  disrupt science and technology}, Nature, 566 (2019), pp.~378--382,
  \url{https://doi.org/10.1038/s41586-019-0941-9},
  \url{https://www.nature.com/articles/s41586-019-0941-9} (accessed
  2021-06-15).

\bibitem{yin_higher-order_2018}
{\sc H.~Yin, A.~R. Benson, and J.~Leskovec}, {\em Higher-order clustering in
  networks}, Physical Review E, 97 (2018), p.~052306,
  \url{https://doi.org/10.1103/PhysRevE.97.052306},
  \url{https://link.aps.org/doi/10.1103/PhysRevE.97.052306} (accessed
  2021-06-15).

\bibitem{yin_local_2017}
{\sc H.~Yin, A.~R. Benson, J.~Leskovec, and D.~F. Gleich}, {\em Local
  {Higher}-{Order} {Graph} {Clustering}}, in Proceedings of the 23rd {ACM}
  {SIGKDD} {International} {Conference} on {Knowledge} {Discovery} and {Data}
  {Mining}, Halifax NS Canada, Aug. 2017, ACM, pp.~555--564,
  \url{https://doi.org/10.1145/3097983.3098069},
  \url{https://dl.acm.org/doi/10.1145/3097983.3098069} (accessed 2021-06-15).

\bibitem{young_hypergraph_2021}
{\sc J.-G. Young, G.~Petri, and T.~P. Peixoto}, {\em Hypergraph reconstruction
  from network data}, arXiv:2008.04948 [physics, stat],  (2021),
  \url{http://arxiv.org/abs/2008.04948} (accessed 2021-06-15).
\newblock arXiv: 2008.04948.

\bibitem{young_construction_2017}
{\sc J.-G. Young, G.~Petri, F.~Vaccarino, and A.~Patania}, {\em Construction of
  and efficient sampling from the simplicial configuration model}, Physical
  Review E, 96 (2017), p.~032312,
  \url{https://doi.org/10.1103/PhysRevE.96.032312},
  \url{https://link.aps.org/doi/10.1103/PhysRevE.96.032312} (accessed
  2021-06-15).

\bibitem{zeng2021fresh}
{\sc A.~Zeng, Y.~Fan, Z.~Di, Y.~Wang, and S.~Havlin}, {\em Fresh teams are
  associated with original and multidisciplinary research}, Nature Human
  Behaviour, 5 (2021), pp.~1314--1322.

\end{thebibliography}
